\documentclass[10pt, final]{IEEEtran}
\usepackage{amsthm}
\usepackage{amssymb}
\usepackage{amsmath}
\usepackage{bbm}
\usepackage[linesnumbered,ruled,vlined]{algorithm2e}
\usepackage{mathrsfs}
\usepackage{cite}
\usepackage{threeparttable,booktabs}
\usepackage{bm,epsfig,amsthm,url}
\usepackage{indentfirst}
\usepackage{amsfonts}
\usepackage{epstopdf}
\usepackage{cases}
\usepackage[caption=false,font=scriptsize]{subfig}
\usepackage{xcolor}
\usepackage{mathtools}

\makeatletter
\renewcommand{\maketag@@@}[1]{\hbox{\m@th\normalsize\normalfont#1}}%
\makeatother
\usepackage{hyperref}
\usepackage{stfloats}
\usepackage{algpseudocode}
\usepackage{booktabs}
\usepackage{hyperref}
\newtheoremstyle{mystyle}{}{}{}{}{}{: }{0pt}{\indent \it{\thmname{#1}\thmnumber{ #2}\thmnote{#3}}}
\theoremstyle{mystyle}
\newtheorem{lemma}{Lemma}

\usepackage[letterpaper, left=0.68in, right=0.68in, bottom=1in, top=0.72in]{geometry}
\def\BibTeX{{\rm B\kern-.05em{\sc i\kern-.025em b}\kern-.08em
    T\kern-.1667em\lower.7ex\hbox{E}\kern-.125emX}}
\allowdisplaybreaks[4]

\begin{document}

\title{Model-Driven Deep Learning Enhanced Joint Beamforming and Mode Switching for RDARS-Aided MIMO Systems}
%\author{Chengwang~Ji,~\IEEEmembership{Member,~IEEE,}
%        John~Doe,~\IEEEmembership{Fellow,~OSA,}
%        and~Jane~Doe,~\IEEEmembership{Life~Fellow,~IEEE}% <-this % stops a space
%\thanks{M. Shell was with the Department
%of Electrical and Computer Engineering, Georgia Institute of Technology, Atlanta,
\author{Chengwang~Ji, Kehui~Li, Haiquan~Lu, Qiaoyan~Peng, Jintao~Wang, Feifei~Gao,~\IEEEmembership{Fellow,~IEEE}, and~Shaodan~Ma,~\IEEEmembership{Senior Member,~IEEE}
\thanks{This article was presented in part at the IEEE ICC 2025\cite{ji_ICC_2025}.

C. Ji, K. Li, H. Lu, Q. Peng, J. Wang and S. Ma are with the State Key Laboratory of Internet of Things for Smart City and the Department of Electrical and Computer Engineering, University of Macau, Macao SAR, China (e-mails: ji.chengwang@connect.um.edu.mo; yc47997@um.edu.mo; haiq\_lu@163.com; qiaoyan.peng@um.edu.mo; wang.jintao@um.edu.mo; shaodanma@um.edu.mo). 
F. Gao is with Department of Automation, Tsinghua
University, Beijing 100084, China (email: feifeigao@ieee.org).}
\vspace{-15pt}
}
\maketitle
\begin{abstract}
Reconfigurable distributed antenna and reflecting
surface (RDARS) is a promising architecture for future sixth-generation (6G) wireless networks. In particular, the dynamic working mode configuration for the RDARS-aided system brings an extra selection gain compared to the existing reconfigurable intelligent surface (RIS)-aided system and distributed antenna system (DAS). In this paper, we consider the RDARS-aided downlink multiple-input multiple-output (MIMO) system and aim to maximize the weighted sum rate (WSR) by jointly optimizing the beamforming matrices at the based station (BS) and RDARS, as well as mode switching matrix at RDARS. The optimization problem is challenging to be solved due to the non-convex objective function and mixed integer binary constraint. To this end, a penalty term-based weight minimum mean square error (PWM) algorithm is proposed by integrating the majorization-minimization (MM) and weight minimum mean square error (WMMSE) methods. 
To further escape the local optimum point in the PWM algorithm, a model-driven DL method is integrated into this algorithm, where the key variables related to the convergence of PWM algorithm are trained to accelerate the convergence speed and improve the system performance. Simulation results are provided to show that the PWM-based beamforming network (PWM-BFNet) can reduce the number of iterations by half and achieve performance improvements of 26.53\% and 103.2\% at the scenarios of high total transmit power and a large number of RDARS transmit elements (TEs), respectively.
\end{abstract}

\section{Introduction}
The demand of enhanced wireless communication performance has driven the development of various critical technologies. For example, by exploiting multiple antennas to harness spatial diversity and multiplexing gains, massive multi-input and multi-output (MIMO) has become the key technology in current fifth-generation (5G) and future sixth-generation (6G) mobile communication networks \cite{xue_survey, haiquan_MIMO}. 
To mitigate the practical issues of millimeter wave (mmWave) and terahertz (THz) high-frequency bands, such as severe path loss and limited material penetration capabilities, distributed antenna systems have been proposed, where multiple geographically distributed antennas coordinate to transmit signals, so as to significantly enhance the spatial multiplexing capabilities in mmWave communications \cite{xue_survey, DAS_TSP, THz_jiang, THz_sota}. However, each distributed antenna requires a dedicated radio-frequency (RF) chain, which leads to high hardware cost and energy consumption \cite{Mohammadi_DAS_high_cost}. To address this issue, reconfigurable intelligent surface (RIS)/intelligent reflecting surface (IRS) has been proposed as a cost-effective and energy-efficient technology. Specifically, RIS is able to create a virtual link to enable the incident signal towards the desired area, whereas its performance suffers from the multiplicative fading effect, thus requiring a large number of reflecting elements to obtain the high reflection gain \cite{YuanweiLiu_RIS, haiquan_RIS, Ahmed_channel_model, ChengzhiMa_ANewArchi, Trinh_Multi_fading_RIS, Zhang_active_RIS, peng_semiRIS}. Besides, the efficient control and fast phase adjustment for RIS are two challenging problems to be solved. On the one hand, the wireless link to control RIS occupies the frequency and space resources. On the other hand, the synchronization problem between the base station (BS) and RIS needs to be solved. In particular, the dedicated wire link can be established to transmit the control and passive phase signals, with reduced latency and resource consumption. Nonetheless, this results in higher hardware costs \cite{Wang_RDARS}.

Recently, reconfigurable distributed antennas and reflecting surface (RDARS) is proposed as a promising technology for 6G mmWave communications \cite{ChengzhiMa_RDARS, ji2025reconfigurable, Wang_RDARS, ChengzhiMa_ANewArchi, jintao_ISAC_RDARS, zhang_RDARS}. By integrating the benefits of the RIS and distributed antennas, RDARS showcases the potential in many aspects, such as improving system capacity \cite{ji2025reconfigurable}, overcoming multiplicative fading \cite{zhang_RDARS}, and enhancing communication reliability \cite{Wang_RDARS}. Specifically, RDARS is a novel type of programmable metasurface composed of reconfigurable elements. The working mode of each element can be dynamically adjusted between the connection mode and reflection mode via the RDARS controller \cite{ChengzhiMa_ANewArchi}. The element working in the connection mode is connected with the BS with a cable or fiber, which can transmit or receive signal as a distributed antenna. On the other hand, the element working in the reflection mode functions as the passive element as in the conventional RIS, which can reflect the incident signal to the desired direction. 
Furthermore, compared to the conventional RIS, the additional selection gain can be achieved via a switching network \cite{Wang_RDARS}.
. 

Due to the promising advantages of RDARS, many research efforts have been devoted to this direction and demonstrated its superiority over RIS-aided systems and DAS via the theoretical performance analysis and prototype experiment \cite{ChengzhiMa_RDARS, ji2025reconfigurable, Wang_RDARS, ChengzhiMa_ANewArchi, jintao_ISAC_RDARS, zhang_RDARS}. 
To reduce the overhead of channel estimation
and computational complexity of the system, the authors in \cite{ChengzhiMa_RDARS} proposed a two-timescale transceiver design.
In \cite{ji2025reconfigurable}, a reconfigurable codebook and low overhead beam training tailored for RDARS were proposed, so as to cater to practical communication scenarios. 
To improve transmission reliability, the mean square error (MSE) was minimized in the uplink MIMO communication system \cite{Wang_RDARS}.
Furthermore, the authors in \cite{ChengzhiMa_ANewArchi} developed a RDARS-aided uplink system prototype, and experiments were conducted to verify the improvement of ergodic achievable rate. 
Besides the wireless communication, an integrated sensing and communication (ISAC) system prototype was developed by leveraging the geometric relationship between BS and RDARS, and a high sensing performance was achieved without incurring communication performance loss \cite{jintao_ISAC_RDARS}. 
In \cite{zhang_RDARS}, the radar output signal-to-noise ratio (SNR) was maximized to improve sensing performance, thanks to the flexible working mode selection brought by RDARS. 

It is worth mentioning that the efficient beamforming design is an important research hotpot in recent years, so as to meet the increasing demand of channel capacity \cite{secure_optimi_Son, WSR_FP_PDD_wang, WSR_FP_BDD_peng, sum_rate_penal_MM_wang, max_min_ADMM_wang}. 
In \cite{secure_optimi_Son}, the secrecy rate was maximized by jointly optimizing hybrid beamforming at
BS and passive beamforming at RIS, where weighted minimum mean-squared error (WMMSE) transformation and
penalty-dual-decomposition (PDD) were utilized to tackle this problem. In \cite{WSR_FP_PDD_wang}, the problem was firstly transformed into its equivalent form by using the fractional programming (FP) method, followed by PDD to maximize the weighted sum rate (WSR) in a double RIS-aided system. By taking both transceiver and RIS hardware impairments into consideration, the WSR maximization problem has been efficiently solved by invoking the block coordinate descent (BCD) framework for an active RIS-assisted system \cite{WSR_FP_BDD_peng}. To investigate the cell-free massive MIMO system performance under both ring and star topologies, a penalty-majorization-minimization (MM)-based distributed beamforming design algorithm was proposed to maximize the achievable sum rate at each BS by decomposing the high-order terms\cite{sum_rate_penal_MM_wang}. Considering the user fairness, the minimum user achievable rate was maximized, and the formulated problem was efficiently solved by a randomized
alternating direction method of multiplier (ADMM) algorithm \cite{max_min_ADMM_wang}. Nonetheless, as the numbers of BS antennas and RIS elements increase, existing beamforming optimization algorithms are faced with high computational complexity, due to the enlarged channel matrix dimension. Besides, these algorithms involve multiple iterations, without guaranteed global optimality, and may lead to a local optimum solution. Moreover, random or fixed initialization points were mainly applied in existing algorithms, which may result in a low-speed convergence behavior.

To tackle these issues, deep learning (DL)-based beamforming design has gained the significant research interest \cite{FL_BF_Ahmet, CNN_BF_shen, DRL_BF_huang, Unsupervised_BF_chen}. In \cite{FL_BF_Ahmet}, the federated learning method was utilized to train the mapping relationship between the channel and beam index, so as to design the beamforming vector. Similarly, parallel convolution neural networks (CNNs) were introduced to reduce the training overhead of beamforming design in \cite{CNN_BF_shen}.
However, the supervised learning relying on the pre-labeled dataset may incur performance degradation. Considering this issue, the sum rate maximization problem was efficiently solved by utilizing deep reinforcement learning (DRL), which embraces the advantages of DL in neural network training and reinforcement learning (RL) algorithms \cite{DRL_BF_huang}. Additionally, an unsupervised learning network was proposed to realize the mapping from channel to beamforming with high computational efficiency in cell-free MIMO systems \cite{Unsupervised_BF_chen}. Note that the aforementioned works are mainly the data-driven methods, which usually have high network training overhead and computational memory requirement. Besides, data-driven methods are difficult to adapt to the complex and dynamic environments, and thus high performance can not be guaranteed. Moreover, data-driven neural networks with multiple hidden layers and the requirement for extensive hyperparameter tuning suffer from limited interpretability.

To overcome these drawbacks, the paradigm of model-driven neural network was proposed by integrating neural networks into the iteration steps in optimization algorithms, or unfolding all iterations into a lay-wise structure \cite{deep_unfolding_alexios, model_driven_he, LowCom_Jie, deep_unfolding_xu, grant_trans_deep_unfolding_sun, deep_unfolding_MO_chen}. 
Specifically, a fixed number of iterations is set to guarantee a fast convergence behavior, unlike the conventional optimization algorithms \cite{model_driven_he}. Moreover, the expert knowledge can be utilized to simplify the network training. For example, a simple solution structure of the active beamforming was introduced in \cite{LowCom_Jie}, which helps to generate a better initialization of WMMSE algorithm. In addition, an approximating matrix inversion was proposed for the MM-based algorithm to reduce the computation complexity in \cite{deep_unfolding_xu}. To reduce the number of trainable parameters, several important parameters related to the specific model are set to be trainable, such as the scaling matrix of matrix inversion operation, the weights of forward-backward splitting, and the step sizes of retraction function in the manifold optimization algorithm \cite{deep_unfolding_xu, grant_trans_deep_unfolding_sun, deep_unfolding_MO_chen}.

Meanwhile, we notice that for a RDARS-aided multi-user communication system, several fundamental issues on WSR maximization remain unsolved. 
First, how to obtain the seamless blend of the optimal beamforming design and RDARS element configurations is still an intractable problem. 
To be specific, one of the challenges in the considered system is the joint design of BS beamforming and RDARS beamforming. Besides, how to achieve effective operation mode configuration remains unknown, due to that the interactively coupled variables. Moreover, the binary mode switching constraint exacerbates the challenges of system design. 
Second, how to obtain the efficient beamforming initialization points? Since inappropriate initialization may result in a low convergence speed and local optimal solution, it is important to develop a proper initialization method for the joint beamforming designs.
Third, how to reduce the computational complexity by embedding the iterations of the proposed optimization algorithm within the neural networks? This problem arises from the high computational complexity suffered by the joint beamforming design, as well as the new mode switching requirement in RDARS systems. Specifically, some key hyper-parameters are critical to intertwine the deep unfolding layer-wise structure with the network to be trained.

In this paper, we investigate the WSR maximization problem in the RDARS-aided mmWave downlink MIMO system, by jointly optimizing the BS beamforming, RDARS beamforming, and mode switching, subject to the total transmit power and binary mode switching constraints. First, we propose a penalty term-based WMMSE (PWM) algorithm. To accelerate the
convergence speed and improve the system performance, a model-driven neural network is next proposed by leveraging expert knowledge related to the active beamforming design and adaptively learning the penalty terms in the mode switching matrix optimization according to the current iteration step in progress.    
The main contributions of this paper are summarized as follows:
\begin{itemize}
    \item  Firstly, we propose an efficient algorithm based on the convex optimization methods for WSR maximization. Specifically, the original optimization problem is equivalently transformed into a more tractable one based on the weighted minimum mean square error (WMMSE)-based algorithm, and the optimal active beamforming is derived in closed-form. Then, the power iteration algorithm is used to meet the unit-modulus constraint in the equivalent problem. Moreover, the binary mode selection constraints are satisfied by introducing the penalty term and majorization minimization (MM) method.
    With the alternative optimization (AO) algorithm, each block of variables is iteratively optimized in an alternate way until convergence is achieved.
    % \item  Secondly, in order to reduce the computational complexity and improve the convergence speed of the above optimization method, we propose a learning-based approach.
    % Specifically, a penalty term-based data generation algorithm with reconfigurable codebook (RCB) is proposed to generate a high-quality dataset. The optimized model switching matrix is obtained by the PWM algorithm, and the labels corresponding to beam indices of active and passive beamforming vectors are then generated by the RCB-based beamforming design algorithm. Subsequently, the data-driven DL is utilized to train the network.
    \item  Secondly, a model-driven DL is integrated into the proposed optimization algorithm to further accelerate the convergence speed and improve the system performance. Specifically, a simple solution structure of active beamforming is introduced for its initialization. The equivalent transmit power and auxiliary terms related to the simple solution structure are set as trainable parameters. Then, an adaptive penalty term is trained by taking the computational complexity of mode switching optimization into consideration. These trainable parameters are trained by the model-driven PWM-based beamforming network (PWM-BFNet).
    \item  Lastly, numerical results demonstrate the superiority of the RDARS architecture, in terms of reducing the number of transmit antennas and transmit power. It is shown that the proposed algorithms are capable of significantly accelerating the convergence speed and increasing the system performance, especially for the case with high transmit power and a large number of RDARS transmit elements (TEs).
\end{itemize}

\textit{Notations:} For a complex vector $\bf x$, $x_i$ represents the $i$-th entry. The elements of vector ${\bf x}[a,b]$ are comprised of the elements of vector ${\bf x}$, beginning with the $a$-th element and ending at the $b$-th element. For a complex matrix $\bf X$, $\mathbf{X}_{[i,j]}$ represents the element in the $i$-th row and $j$-th column. $\operatorname{Tr}(\bf X)$, $||\bf X||$ and $||\mathbf X||_{F}$ denote its trace, the 2-norm and F-norm of $\bf X$, respectively. ${\bf x}^{T}$, ${\bf x}^{ H}$, ${\bf X}^{T}$ and ${\bf X}^{H}$ stand for the transpose and conjugate transpose of vector $\bf x$ and matrix $\bf X$, respectively. 
\section{System Model and Problem Formulation} \label{sec: system model}

\subsection{System Model}
Fig. \ref{fig: system architecture} shows a RDARS-aided MIMO system with a BS and $K$ single-antenna users. The BS is equipped with $N_{\rm{t}}$ antennas, and the RDARS has $N$ elements. 
Each element of RDARS can work in two modes: \textit{connection mode} and \textit{reflection mode}. 
The working mode of each element is determined by a diagonal mode switching matrix ${\bf{A}}\in \mathbb{R}^{N\times N}$ with $\mathbf{A}_{[i, i]} \in \{0,1\}$.
%, where $\bf{A}$ is a $N \times N$ diagonal matrix and $\mathbf{A}_{[i,i]} \in \{0,1\}$ for the $i$-th element. 
To be specific, when $\mathbf{A}_{[i, i]} = 1$, the $i$-th element works in the connection mode and can be regarded as a distributed antenna with the capabilities of transmitting and receiving signal.
Otherwise, $\mathbf{A}_{[i,i]} = 0$ represents that the $i$-th element operates in the reflection mode which functions as a passive element to reflect the incident signals. Let $a$ denote the number of elements working in the connection mode, while the remaining $N-a$ elements operate in the reflection mode.
In other words, $a$ RDARS elements are programmed as the distributed antennas like transmit elements, so as to serve multiple users cooperatively with BS antennas. 
Let $\mathcal{N_{\rm{t}}}=\{1,2,\dots, N_{\rm{t}}\}$, $\mathcal{N}=\{1,2,\dots, N\}$, $\mathcal{A}=\{1,2,\dots, a\}$ and $\mathcal{K}=\{1,2,\dots, K\}$ denote the sets of indices of BS antennas, RDARS elements, connected elements, and users, respectively.

By denoting the transmit symbol vector as $\mathbf{s} \in  \mathbb{C}^{K \times 1}$ with ${\mathbb{E}\{\mathbf{s}\mathbf{s}^H \}}={\bf{I}}_{K}$, the combined transmit signal from the BS and RDARS elements in the connection mode is
\begin{align}\label{equ: transmit signal }
 \mathbf{x} &= 
 \begin{bmatrix}
{\mathbf{W}}_{\rm{b}} \\
{\mathbf{W}}_{\rm{r}} 
\end{bmatrix}
\mathbf{s} =\mathbf{F} \mathbf{s},
\end{align}
where ${\mathbf{W}}_{\rm{b}}=\left[\mathbf{w}_{{\rm{b}},1},\mathbf{w}_{{\rm{b}},2}, \cdots,\mathbf{w}_{{\rm{b}},K}\right] \in \mathbb{C}^{N_{\rm{t}} \times K}$ and $\mathbf{W}_{\rm{r}} = \left[ \mathbf{w}_{{\rm{r}},1}, \mathbf{w}_{{\rm{r}},2},\cdots, \mathbf{w}_{{\rm{r}},K}\right] \in \mathbb{C}^{a \times K}$ denote the BS and RDARS beamforming matrices, respectively. 
$\mathbf{F} = \left[ \mathbf{f}_{1}, \mathbf{f}_{2},\cdots, \mathbf{f}_{K}\right] \in \mathbb{C}^{(N_{\rm{t}}+a) \times K}$ represents the equivalent transmit beamforming matrix with ${\bf{f}}_k \in \mathcal{C}^{(N_{\rm{t}}+a) \times 1}$.
\begin{figure}[t] 
 \centering
 \includegraphics[width=0.4\textwidth]{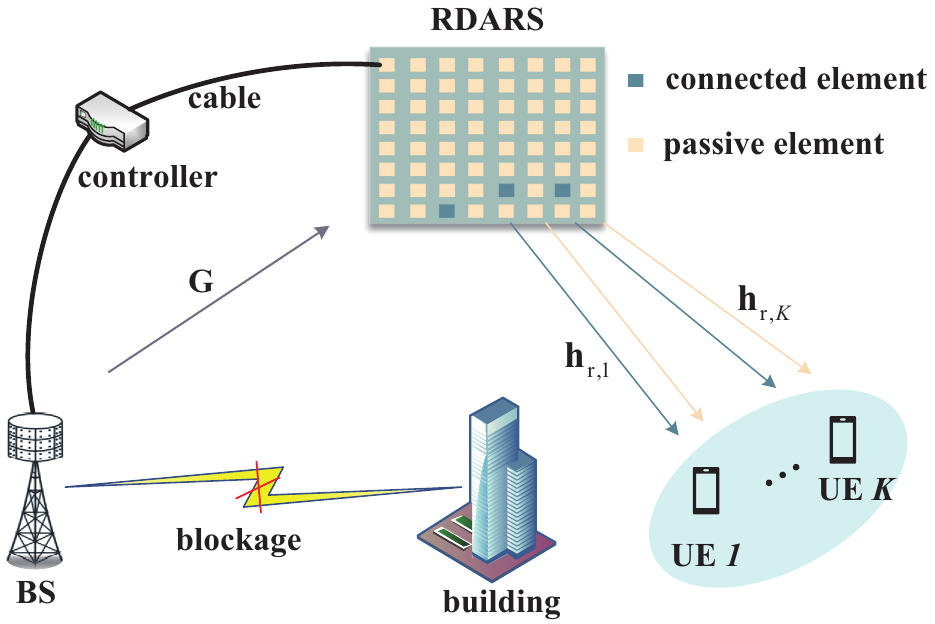}
 \caption{A RDARS-assisted downlink MIMO architecture.}
 % \vspace{-13pt}
 \label{fig: system architecture}
 \end{figure}
The direct links between the BS and user equipments (UEs) are assumed to be severely blocked by obstacles, due to the high directivity of mmWave signals. 
The channels from the BS to RDARS and that from RDARS to the $k$-th UE are denoted by ${{\bf{G}}} \in \mathbb{C}^{N \times N_{\rm{t}}}$ and ${\bf{h}}_{{\rm{r}},k} \in \mathbb{C}^{N \times 1}$, respectively\footnote{The channel estimation (CSI) can be obtained based on the flexible mode switching of RDARS elements. Specifically, the individual CSI of the BS-RDARS channel and RDARS-UE channel, can be obtained by switching the working modes of connected elements as the transmit or receive antennas \cite{ji2025reconfigurable, Wang_RDARS, ChengzhiMa_RDARS}. To reduce the overhead, a practical channel estimation scheme based on the linear minimum mean square error (LMMSE) criteria has been proposed in \cite{ChengzhiMa_RDARS}, where only a limited number of pilots is required for this channel estimator, thus appealing for practical applications.}. 
By considering the general Rician channel model, the channels ${{\bf{G}}}$ and ${\bf{h}}_{{\rm{r}},k}$ are given by 
\begin{align}
{\bf{G}} = &\kappa_{{\rm{b}}}\left(\sqrt{\frac{\xi}{\xi+1}} \tilde{{\bf{b}}}(N, \chi,\psi){\bf{b}}^{H}({{N_{\rm{t}}}},\theta) + \sqrt{\frac{1}{\xi+1}} \widetilde{{\bf{G}}}\right), \label{equ: Channel BR}\\
 {\bf{h}}_{{\rm{r}},k} =& \kappa_{{\rm{r}},k} \left(\sqrt{\frac{\xi}{\xi+1}}\tilde{\bf{b}}(N, \phi_{k},\upsilon_{k}) + \sqrt{\frac{1}{\xi+1}}\widetilde{{\bf{h}}}_{r,k}
 \right),\label{equ: Channel r}
\end{align}
where $\xi$ denotes the Rician factor, and $\kappa_{{\rm{b}}}$ and $\kappa_{\rm{r},k}$ denote the complex-valued channel coefficients of BS-RDARS and RDARS-UE $k$ channels, respectively.
Let $\chi$ and $\psi$ represent the vertical and horizontal angle-of-arrivals (AoAs) from the BS to RDARS, respectively, $\theta$ represent the angle-of-departure (AoD) from the BS to RDARS. The vertical and horizontal AoAs from the RDARS to the $k$-th UE are denoted by $\phi_{k}$ and $\upsilon_{k}$, respectively. The elements within the non-line-of-sight (NLoS) components and $\widetilde{\bf{h}}_{r,k}$ are characterized by a standard complex Gaussian distribution.
Moreover, the steering vector is
\begin{equation}
    {\bf{b}}(N, \theta ) = \frac{1}{\sqrt{N}}[e^{j\pi\theta \cdot 0}, e^{j\pi\theta \cdot 1}, \cdots , e^{j{\pi}\theta \cdot (N-1)}]^{T},
\end{equation}
where the inter-antenna spacing is half wavelength and $\theta$ denotes the spatial frequency.
Then, we have $\tilde{\bf{b}}(N, \chi,\psi ) = {\bf{b}}(N_{\rm{z}}, \chi ) \otimes {{\bf{b}}}(N_{\rm{y}}, \psi )$, where $N_{\rm{z}}$ and $N_{\rm{y}}$ denote the numbers of elements along with the vertical and horizontal directions, respectively. 

The received signal for the $k$-th UE is given by
\begin{equation}\label{equ: y_k}
{y_k} = {\bf{h}}_{{\rm{r}},k}^H {\bf{(}}{{\bf{I}}_N}{\bf{ - A)\Phi }} {{\bf{G}}}{\bf{W}}_{{\rm{b}}} {\bf{s}} + {\bf{h}}_{{\rm{r}},k}^H{\bf{\tilde A}} {\bf{W}}_{{\rm{r}}} {\bf{s}} + {n_k},
\end{equation}
where ${\bf{\Phi}}\in \mathbb{C}^{N \times N}$ denotes the RDARS's phase shift matrix, mode switching matrix ${\bf{\tilde A}} \in \mathbb{C}^{N \times a}$ consists of non-zero columns of the sparse mode switching matrix $\bf{A}$, and $n_k$ denotes the additive white Gaussian noise (AWGN) with power $\sigma^2_k$. 
%We denote ${\bf{s}}$ as the transmit symbol vector with $\mathbb{E}\{s_i s_i\} = 1$ and $\mathbb{E}\{s_i s_j\} = 0$, where $\mathbb{E}\{ \cdot \}$ denote the expectation function.
Let $\bm{\varphi} = {[{\varphi _1},...,{\varphi _N}]^H} = \operatorname{diag}(\bf{\Phi })$ and ${{\bf{H}}_{\mathrm{r},k}} =\operatorname{diag}({\bf{h}}_{{\rm{r}},k}^H){\bf{G}}$, the received signal can be equivalently expressed as 
\begin{equation}
    y_k = {{\bf{h}}_k}{{\bf{f}}_{k}}{s_k} + \sum\nolimits_{m \ne k}^K {{{\bf{h}}_k}{{\bf{f}}_{m}}{s_m}}  + {n_k},
\end{equation}
where ${{\bf{h}}_k}$ denotes the effective channel between the BS and the $k$-th UE, given by ${{\bf{h}}_k} = [{{\bm{\varphi }}^H} {\bf{(}}{{\bf{I}}_N}{\bf{ - A) }} {{\bf{H}}_{{\rm{r}},k}},{\bf{h}}_{{\rm{r}},k}^H{\bf{\tilde A}}]$.
%
%Then, we denote $\bm{\varphi} = {[{\varphi _1},...,{\varphi _N}]^H} = \operatorname{vec}(\bf{\Phi })$, ${{\bf{H}}_{{\rm{b}},k}} = {\bf{(}}{{\bf{I}}_N}{\bf{ - A)}}{{\bf{H}}_{\mathrm{r},k}}$ with ${{\bf{H}}_{\mathrm{r},k}} =\operatorname{diag}({\bf{h}}_{{\rm{r}},k}^H){\bf{G}}$, and ${{\bf{h}}_k} = [{{\bf{\varphi }}^H}{{\bf{H}}_{{\rm{b}},k}},{\bf{h}}_{{\rm{r}},k}^H{\bf{\tilde A}}]$. Then, ${y_k}$ is rewritten as
%
The signal-to-noise-plus-interference ratio (SINR) for the $k$-th UE is 
\begin{equation}\label{equ: SINR}
   {\gamma_k} = \frac{{{{\left| {{{\bf{h}}_k}{{\bf{f}}_{k}}} \right|}^2}}}{{{{\sum\nolimits_{m \ne k}^K |{{{\bf{h}}_k}{{\bf{f}}_{m}}} }|^2} + {\sigma _k}^2}}.
\end{equation}
The achievable rate of UE k is 
${R_k} = {\log _2}(1 + {\gamma _k})$.
\subsection{Problem Formulation}\label{sec: problem formulation}
Our objective is to maximize the WSR, by jointly optimizing the BS beamforming, RDARS beamforming, and mode switching matrix.
The optimization problem can be formulated as
\begin{subequations}\label{pro: weighted sum rate}
 \begin{align}
 \mathop {\max }\limits_{\substack{{{\bf{F}}},{\bf{\Phi }},{\bf{A}}, {\tilde{\bf{A}}}}} 
 & \;\;\sum\nolimits_{k = 1}^K {{\alpha _k}{R_k}}
 \\
 \;\textrm{s.t.}\;
 & ~ \operatorname{Tr}({{\bf{F}}}{\bf{F}}^H) \le {P_{\rm{tot}}} ,\label{con: P}\\
 & ~|\mathbf{\Phi}_{[i,i]}| = 1,  \forall i\in{\mathcal{N}},\label{con: Phi} \\
 & ~\sum\nolimits_{i = 1}^{N}\!\mathbf{A}_{[i,i]} \!\!= a, \mathbf{A}_{[i,i]} \!\in\! \{0,1\}, \forall i\in{\mathcal{N}}, \label{con: A}\\
 & ~\sum\nolimits_{i = 1}^{N}\!\tilde{\mathbf{A}}_{[i,l]} \!\!= 1, \tilde{\mathbf{A}}_{[i,l]} \!\!\in\!\! \{0,1\} , \forall l\in{\mathcal{A}}, i\in{\mathcal{N}},\label{con: A tilde}\\
 & ~ {\bf{A} }= \tilde{{\bf{A}}}\tilde{{\bf{A}}}^{H},\label{con: A +A tilde}
 \end{align}
\end{subequations}
where $\alpha_k$ denotes the weighted factor of the $k$-th UE, and $P_{\rm{tot}}$ denotes the maximum total transit power. 
It is observed from \eqref{pro: weighted sum rate} that it is a mixed-integer nonlinear programming (MINLP) problem, which is a non-convex problem due to highly coupled variables and unit-modulus constraints. Moreover, the binary mode switching constraints exacerbate the challenge. To solve this problem, we propose an efficient penalty term-based WMMSE (PWM) algorithm in the following section.

\section{PWM Algorithm} \label{sec: PWM algorithm}
In this section, we first reformulate the original WSR problem via the WMMSE method \cite{Qi_WMMSE_Algo}. Then, an AO-based algorithm is proposed to solve the reformulated problem based on MM and penalty methods.

By temporarily removing constraint \eqref{con: P} and introducing the auxiliary vectors $\bm{\lambda}=[\lambda_1, \cdots, \lambda_K]^{T}$ and ${\bf{u}} = [u_1, \cdots, u_K]^{T}$,
problem \eqref{pro: weighted sum rate} can  be transformed into a WMMSE problem as 
%\vspace{-8pt}
\begin{align}\label{pro: weighted sum rate MMSE}
 \mathop {\min }\limits_{\substack{{{\bf{F}}},{\bf{\Phi }},{\bf{A}},\\ {\tilde{\bf{A}}}, \bf{u}, \bm{\lambda}}}\;
 \sum\limits_{k = 1}^K {{\alpha _k}({\lambda _k}{e_k} - \log {\lambda _k})}
\;\;\textrm{s.t.}\;
\eqref{con: Phi}, \eqref{con: A}, \eqref{con: A tilde}, \eqref{con: A +A tilde},
\end{align}
where $e_k$ denotes the MSE for detecting UE $k$'s signal, which is given by 
\begin{align}\label{equ: WMMSE}
{e_k} &= 1 - u_k^H{{\bf{h}}_k}{{\bf{f}}_{k}} -{\bf{f}}_{k}^H{\bf{h}}_k^H{u_k} \notag \\
&+ u_k^H{{\bf{h}}_k}\sum\limits_{m = 1}^K {{{\bf{f}}_{m}}{\bf{f}}_{m}^H} {\bf{h}}_k^H{u_k} + u_k^H{u_k}\frac{{\sigma_{k}^{2}}}{P_{\rm{tot}}}\sum\limits_{m = 1}^K {{\bf{f}}_{m}^H{{\bf{f}}_{m}}}.
\end{align}
The following Lemma is given to verify the equivalence between problem \eqref{pro: weighted sum rate} and \eqref{pro: weighted sum rate MMSE}.
\begin{lemma}\label{Lemma: WMMSE equivalent pro}
  Problem \eqref{pro: weighted sum rate MMSE} is equivalent to the weight sum rate maximization problem \eqref{pro: weighted sum rate}, where the global optimal solutions for the two problems are identical.
\end{lemma}
\begin{proof}
  Please refer to Appendix A.
\end{proof}
Note that problem \eqref{pro: weighted sum rate MMSE} is still a non-convex problem due to the coupled variables. 
To this end, the AO-based method is applied, where all the variables are divided into four blocks, i.e., 1) active beamforming $\{ {\bf{f}}_{k}, u_k, \lambda_k\}$, 2) passive beamforming $\bf{\Phi}$, 3) sparse mode switching matrix ${\bf{A}}$ and 4) mode switching matrix $\tilde{\bf{A}}$. Then each block of variables is optimized in an iterative manner until convergence is achieved.

\subsection{Active Beamforming Optimization}
With fixed $\bf{\Phi}$, ${\bf{A}}$ and $\tilde{{\bf{A}}}$, problem \eqref{pro: weighted sum rate MMSE} is reduced to 
\begin{align}\label{pro: weighted sum rate MMSE fixed Phi, A}
 \mathop {\min }\limits_{{{\bf{F}}},\bf{u}, \bm{\lambda}} \quad
 & \sum\nolimits_{k = 1}^K {{\alpha _k}({\lambda _k}{e_k} - \log {\lambda _k})}.
\end{align}
The solutions to problem \eqref{pro: weighted sum rate MMSE fixed Phi, A} can be derived in closed-form as
\begin{align}\label{equ: opt_lammda}
\lambda _k^{\mathrm{opt}} = &{e_{k}^{-1}},\\
u_k^{\mathrm{opt}} = &{({{\bf{h}}_k}\sum\limits_{m = 1}^K {{{\bf{f}}_{m}}{\bf{f}}_{m}^H} {\bf{h}}_k^H + \frac{{\sigma_{k}^{2}}}{P_{\rm{tot}}}\sum\limits_{m = 1}^K {{\bf{f}}_{m}^H} {{\bf{f}}_{m}})^{ - 1}}{{\bf{h}}_k}{{\bf{f}}_{k}}, \label{equ: opt_u}\\
{\bf{f}}_{k}^{\mathrm{opt}}=& {\alpha _k}{u_k}{\lambda _k}(\!\sum\limits_{m = 1}^K {{\alpha _m}u_m^H{u_m}{\lambda _m}(\frac{{\sigma_{m}^{2}}}{P_{\rm{tot}}}{{\bf{I}}_{N + a}}\!\! + \!\!{\bf{h}}_k^H{{\bf{h}}_k}){)^{ - 1}}{\bf{h}}_k^H}. \label{equ: opt_f_R}
\end{align}
By substituting \eqref{equ: opt_u} into \eqref{equ: WMMSE}, we obtain the minimum MSE (MMSE) as
$e_k^{\mathrm{mmse}} = 1 - {\bf{f}}_{k}^H{\bf{h}}_k^H{J_{k}^{-1}}{{\bf{h}}_k}{{\bf{f}}_{k}}$ 
, with ${J_k}={{\bf{h}}_k}\sum\nolimits_{m = 1}^K {{{\bf{f}}_{m}}{\bf{f}}_{m}^H} {\bf{h}}_k^H + \frac{\sigma_{k}^{2}}{P_{\rm{tot}}}\sum\nolimits_{m = 1}^K {{\bf{f}}_{m}^H} {{\bf{f}}_{m}}$.
\subsection{Passive Beamforming Optimization}
Given fixed $u_k$, $\lambda_k$, ${\bf{f}}_{k}$, ${\bf{A}}$ and $\tilde{{\bf{A}}}$, problem \eqref{pro: weighted sum rate MMSE} is reduced to
\begin{align}
 \mathop {\min }\limits_{\bm{\varphi}} \quad {{\bm{\varphi }}^H}\bf{C}{\bm{\varphi }} + {{\bm{\beta }}^H}{\bm{\varphi }} + {{\bm{\varphi }}^H}{\bm{\beta }}
 \;\;\;\;\;\;\;\textrm{s.t.}\;\; \eqref{con: Phi},\label{pro: passive BF}
 \end{align}
where ${\bf{C}} = \sum\nolimits_{k = 1}^{K} {\alpha _k}{\lambda _k}u_k^H{u_k}({\bf{I}}_{N}- {\bf{A}}){{\bf{H}}_{{\rm{r}},k}}\sum\nolimits_{m = 1}^K {{\bf{w}}_{\mathrm{b},m}}{\bf{w}}_{\mathrm{b},m}^H {{\bf{H}}^{H}_{{\rm{r}},k}} ({\bf{I}}_{N}- {\bf{A}})$ and ${\bm{\beta }} = \sum\nolimits_{k = 1}^K { {\alpha _k}{\lambda _k}u_k^H{u_k}({\bf{I}}_{N}- {\bf{A}}){{\bf{H}}_{{\rm{r}},k}}\sum\nolimits_{m = 1}^K {{{\bf{w}}_{\mathrm{b},m}}{{\bf{w}}^H_{\mathrm{r},m}}{{{\bf{\tilde A}}}^H}} } {\bf{h}}_{{\rm{r}},k}- {\alpha _k}{\lambda _k}u_k^H({\bf{I}}_{N}- {\bf{A}}){{\bf{H}}_{{\rm{r}},k}}{{\bf{w}}_{\mathrm{b},k}}$.
Let $\mathbf{p}= [\bm{\varphi} , q]^T$ where $p_n$ is the $n$-th element of ${\bf{p}}$, and $q$ is an auxiliary variable. Therefore, problem \eqref{pro: passive BF} is equivalent to
\begin{align}\label{pro: rankone}
 \mathop {\max}\limits_{\bf{p}} 
 \;\;{{\bf{p}}^H}{\bf{Dp}}
 \;\;\textrm{s.t.}|p_n| = 1, n= 1,\cdots, N+1,
\end{align}
where ${\bf{D}} = \left[ { - {\bf{C}}}, { - {\bm{\beta}}};{- {{\bm{\beta }}^H}}, 0 \right]$.

It is observed that the optimization problem \eqref{pro: rankone} is a unimodular quadratic program and can be solved by the power iteration algorithm. Specifically, the value of $p$ in the $q$-th iteration is \cite{PI_2014}
\begin{equation}\label{equ: power iteration}
  {\bf{p}}^{(q+1)}= e^{j\mathrm{arg}(({\bf{D}}+\varepsilon {\bf{I}}_{N+1}){\bf{p}}^{(q)})},
\end{equation}
where $\varepsilon {\bf{I}}_{N+1}$ is introduced to ensure that ${\bf{D}}+\varepsilon {\bf{I}}_{N+1}$ is positive definite. After the iteration converges, the passive beamforming vector can be derived as $\bm{\varphi} = e^{j \arg(\frac{\mathbf{p}{[1:N]}}{p_{N+1}} )}$, thus solving problem \eqref{pro: passive BF}.
This guarantees the convergence of the iterative algorithm, as shown in Lemma \ref{lemma: Power iteration equivalent}.
\begin{lemma}\label{lemma: Power iteration equivalent}
The problem is equivalent to problem \eqref{pro: rankone} as follows:
\begin{subequations}\label{pro: rankone of D_pie}
\begin{align}
  \mathop {\max}\limits_{\bf{p}} 
\quad&{{\bf{p}}^H}{\bf{D}}^{'} {\bf{p}}\\
\textrm{s.t.}\quad &|p_n| = 1, n= 1,\cdots, N+1.
 \end{align}
\end{subequations}
 Furthermore, the power iteration algorithm is guaranteed to converge to at least a local optimum of problem \eqref{pro: rankone} when ${\bf{D}}^{'}$ is positive-defined.
\end{lemma}
\begin{proof}
  Please refer to Appendix B.
\end{proof}

% Considering the SDR technique, problem (\ref{pro: rankone}) is reformulated as
%  \begin{align}
%  \label{pro: SDR}
%  \mathop {\max }\limits_{\bf{P}} \;\;
% \operatorname{Tr}(\mathbf{P}{\bf{D}}) \;\;\;\;\; \textrm{s.t.} \;\; \mathbf{P}_{[n,n]} = 1, \mathbf{P} \succeq 0.
% \end{align}
% As problem (\ref{pro: SDR}) is a standard convex semi-definite program (SDP), it can be optimally solved by existing convex optimization solvers such as CVX. Then, we can construct a rank-one solution from the obtained higher-rank solution to problem (\ref{pro: SDR}) by Gaussian random vectors, i.e., $\bm{\varphi} = e^{j \arg(\frac{\mathbf{p}{[1:N]}}{p_{N+1}} )}$, where $\mathbf{p}{[1:N]}$ denotes the first $N$ terms of $\bf{p}$.
\subsection{Optimization of Sparse Mode Switching Matrix ${\bf{A}}$}
To create a unified framework, we employ the penalty technique to consolidate constraint \eqref{con: A +A tilde} into penalty terms. 
With fixed $\bf{F}$,$\bf{u}$, $\bm{\lambda}$, $\bf{\Phi}$, and $\tilde{\bf{A}}$, problem \eqref{pro: weighted sum rate} can be reformulated as 
\begin{align}\label{pro: MMSE A}
 \mathop {\min }\limits_{{\bf{A}}}
\sum\limits_{k = 1}^K {{\alpha _k}({\lambda _k}{e_k} \!\!- \!\!\log {\lambda _k})} \!\!+\!\! \frac{1}{2\rho} || \mathbf{A}\!\! - \!\!\tilde{\bf{A}}\tilde{\bf{A}}^{H}||_{F}^{2}\;\;\;\textrm{s.t.}\;
 \eqref{con: A},
 \end{align}
where $1/(2\rho)$ is the penalty factor.

Let $\bf{a} = \operatorname{diag} ( \bf{A} )$ and $\tilde{\mathbf{a}} = [\mathbf{\tilde A}[:,1]^{T}, \cdots, \mathbf{\tilde A}[:,a]^{T}]^{T} =  [\bar{\mathbf{a}}_{1}^{T}, \cdots, \bar{\mathbf{a}}_{a}^{T}]^{T} \in \mathbb{C}^{Na\times 1}$, where $\mathbf{\tilde A}[:,l]$ denotes the $l$-th non-zero column of $\bf{\tilde A}$.
As such, the objective function of problem \eqref{pro: MMSE A} can be rewritten as
\begin{align} \label{pro: A}
    {f_3}({\bf{A}})  &= {\bf{r}}_1^H{\bf{a}} 
    \!+\! {{\bf{a}}^T}{{\bf{r}}_1}\!+\!{{\bf{a}}^T}{{\bf{r}}_2} + {\bf{r}}_2^H{\bf{a}} + {{\bf{a}}^T}{{\bf{R}}_1}{\bf{a}} \nonumber\\
    & + {\bf{r}}_3^H{\bf{a}} \!+\! {{\bf{a}}^T}{{\bf{r}}_3}\!
   + \!\frac{1}{{2\rho }}({\bf{r}}_4^T{\bf{a}}\! +\!{r_5}),
\end{align}
where the auxiliary parameters are given by 
\begin{subequations}\label{equ: auxiliary parameters of A}
\begin{align}
{{\bf{r}}_1} \!\!= & \!\! \sum\nolimits_{k = 1}^K {{\alpha _k}{\lambda _k}{u_k}\operatorname{diag}({{\bf{w}}^H_{{\rm{b}},k}}{{\bf{H}}_{{\rm{r}},k}^H}){\bm{\varphi}}}, \\
{{\bf{r}}_2} \!\!=& \!\!-\! \!\sum\limits_{k = 1}^K {{\alpha _k}{\lambda _k}u_k^H{u_k}{{\bf{\Phi }}^H}{{\bf{H}}_{{\rm{r}},k}}\!\! \sum\limits_{m = 1}^K {{{\bf{w}}_{{\rm{b}},m}}{\bf{w}}_{{\rm{b}},m}^H} {\bf{H}}_{{\rm{r}},k}^H{\bm{\varphi }}},\\
{{\bf{r}}_3}\!\! = & \!\! - \!\! \sum\limits_{k = 1}^K {{\alpha _k}{\lambda _k}u_k^H{u_k}{{\bf{\Phi }}^H}{{\bf{H}}_{{\rm{r}},k}} \!\! \sum\limits_{m = 1}^K {{{\bf{w}}_{{\rm{b}},m}}{\bf{w}}_{{\rm{r}},m}^H} {{{\bf{\tilde A}}}^H}{{\bf{h}}_{{\rm{r}},k}}} , \\
{{\bf{r}}_4} \!\!= & \operatorname{diag}({{\bf{I}}_N} - 2{\bf{\tilde A}}{{{\bf{\tilde A}}}^H}), \\
{r_5} \!\!= & \operatorname{Tr}({\bf{\tilde A}}{{{\bf{\tilde A}}}^H}), \\
{{\bf{R}}_1} \!\!=& \!\! \sum\limits_{k = 1}^K {{\alpha _k}{\lambda _k}u_k^H{u_k}{{\bf{\Phi }}^H}{{\bf{H}}_{{\rm{r}},k}}\!\! \sum\limits_{m = 1}^K {{{\bf{w}}_{{\rm{b}},m}}{\bf{w}}_{{\rm{b}},m}^H} {\bf{H}}_{{\rm{r}},k}^H{\bf{\Phi }}}.
\end{align}
\end{subequations}
The derivations of auxiliary parameters are shown in Appendix C. 

In the following, we first apply the MM technique to find a tractable surrogate function of \eqref{pro: MMSE A}. Based on the second-order Taylor expansion, an upper bound of the term ${{\bf{a}}^T}{{\bf{R}}_1}{\bf{a}}$ can be derived as
$ {{\bf{a}}^T}{{\bf{R}}_1}{\bf{a}} \le {{\bf{a}}^T}{{\bf{\Lambda }}_1}{\bf{a}} + 2\Re \left\{ {{{\bf{a}}^T}({{\bf{R}}_1} - {{\bf{\Lambda }}_1}){{\bf{a}}_t}} \right\} + {{\bf{a}}_t}^T({{\bf{\Lambda }}_1} - {{\bf{R}}_1}){{\bf{a}}_t}$,
where ${{\bf{\Lambda }}_1} = {{\bar \lambda }_{\max }}({{\bf{R}}_1}){{\bf{I}}_N}$, with ${{\bar \lambda }_{\max }}({{\bf{R}}_1})$ denoting the maximum eigenvalue of ${{\bf{R}}_1}$. Then, the surrogate function with respect to (w.r.t.) $\bf{A}$ is given by
${{\bar f}_3}({\bf{A}})= 2\Re \left\{ {{\bf{r}}_1^H{\bf{a}}} \right\} + 2\Re \left\{ {{\bf{r}}_2^H{\bf{a}}} \right\} + 2\Re \left\{ {{\bf{r}}_3^H{\bf{a}}} \right\}
+2\Re \left\{ {{{\bf{a}}^T}({{\bf{R}}_1} - {{\bf{\Lambda }}_1}){{\bf{a}}_t}} \right\} + \frac{1}{{2\rho }}{\bf{r}}_4^T{\bf{a}}$.
Thus, problem \eqref{pro: MMSE A} is rewritten as
\begin{align}\label{pro: a}
 \mathop {\min }\limits_{\bf{a}} \quad
\Re \left\{ {{\bf{r}}_6^H{\bf{a}}} \right\} \quad\  \textrm{s.t.}\quad
 {a_i} \in \left\{ {0,1} \right\},
\end{align}
with ${{\bf{r}}_6} =2{{\bf{r}}_1} + 2{{\bf{r}}_2} + 2{{\bf{r}}_3} + 2({{\bf{R}}_1} - {{\bf{\Lambda }}_1}){{\bf{a}}_t} + \frac{1}{{2\rho }}{{\bf{r}}_4}$.
Let $\mathcal{M}$ denote the set of first $a$ minimum elements of $\Re\left\{{{\bf{r}}_6}\right\}$. Accordingly, the optimal solution to problem \eqref{pro: MMSE A} can be derived as
\begin{equation} \label{equ: opt_a}
{a}^{\mathrm{opt}}_{i} =
\begin{cases}
 1, &\Re\left\{{{\bf{r}}_6}\right\}_{[i]} \in \mathcal{M} , \\
 0, &\Re\left\{{{\bf{r}}_6}\right\}_{[i]}  \notin \mathcal{M}. 
\end{cases}
\end{equation}
It is observed from \eqref{equ: opt_a} that we need to find the 
first $a$ minimum elements from $\Re \left\{ {\bf{r}}_6\right\}$, i.e., ${a}_{i} = 1$, where the corresponding indexes are the locations of elements working in the connection mode.
\subsection{Optimization of Mode Switching Matrix $\tilde{\bf{A}}$}
By employing the penalty term and for fixed $\bf{F}$,$\bf{u}$, $\bm{\lambda}$, $\bf{\Phi}$ and ${\bf{A}}$, problem \eqref{pro: weighted sum rate} can be reformulated as 
\begin{align}\label{pro: MMSE A + A tilde}
 \mathop {\min }\limits_{\tilde{{\bf{A}}}}
\sum\limits_{k = 1}^K {{\alpha _k}({\lambda _k}{e_k} \!\!- \!\!\log {\lambda _k})} \!\!+\!\! \frac{1}{2\rho} || \mathbf{A}\!\! - \!\!\tilde{\bf{A}}\tilde{\bf{A}}^{H}||_{F}^{2}\;\;\;\textrm{s.t.}
 \eqref{con: A tilde},
\end{align}
where $1/(2\rho)$ is the penalty factor.
With fixed ${\mathbf{a}}$, 
the objective function can be expressed in terms of  $\tilde{\mathbf{a}}$ is 
\begin{equation} \label{pro: tilde A}
   {f_4}({\bf{\tilde A}}) = {{{\bf{\tilde r}}}_1}^H{\bf{\tilde a}} + {{{\bf{\tilde a}}}^H}{{{\bf{\tilde r}}}_1} + {{{\bf{\tilde r}}}_2}^H{\bf{\tilde a}} + {{{\bf{\tilde a}}}^H}{{{\bf{\tilde r}}}_2} + {{{\bf{\tilde a}}}^H}{{\bf{R}}_2}{\bf{\tilde a}},
\end{equation}
where
${{{\bf{\tilde r}}}_1} \!\!= \!\!\sum\limits_{k = 1}^K {{\alpha _k}{\lambda _k}} |u_k|^2\sum\limits_{m = 1}^K {(({{\bf{w}}_{{\rm{r}},m}}{{\bf{w}}^H_{{\rm{b}},m}}{{\bf{H}}^H_{{\rm{r}},k}}{\bf{(}}{{\bf{I}}_N}-{\bf{ A){\bm{\varphi}} }})}^* \!\!\otimes \!\!{{\bf{h}}_{{\rm{r}},k}})$, ${{\bf{R}}_2} \!\!=\!\! \sum\limits_{k = 1}^K {{\alpha _k}{\lambda _k}} \!\!\sum\limits_{m = 1}^K {({{\bf{w}}^{*}_{{\rm{r}},m}}}\! \otimes \!{u_k}{{\bf{h}}_{{\rm{r}},k}})({{\bf{w}}^{T}_{{\rm{r}},m}} \!\otimes \!{{\bf{h}}^{H}_{{\rm{r}},k}}{u}_{k}^{H}) \!- \!\frac{1}{\rho }({{\bf{I}}_a} \!\otimes\!{\bf{A}})$ and ${{{\bf{\tilde r}}}_2} =  - \sum\limits_{k = 1}^K {{\alpha _k}{\lambda _k}} u_k^H({{\bf{w}}^*_{{\rm{r}},k}} \otimes {{\bf{h}}_{{\rm{r}},k}})$.
Based on the second-order Taylor expansion, a convex surrogate function of ${{{\bf{\tilde a}}}^T}{{\bf{R}}_2}{\bf{\tilde a}}$ can be obtained as
\begin{align}
   {{{\bf{\tilde a}}}^T}{{\bf{R}}_2}{{\bf{\tilde a}}}  \le &{{{\bf{\tilde a}}}^T}{{\bf{\Lambda }}_2}{\bf{\tilde a}} + 2\Re \left\{ {{{{\bf{\tilde a}}}^T}({{\bf{R}}_2} - {{\bf{\Lambda }}_2})}\mathbf{\tilde a}_{t^{'}}\right\} \nonumber\\&+ \mathbf{\tilde a}_{t^{'}} ^{T}({{\bf{\Lambda }}_2} - {{\bf{R}}_2})\mathbf{\tilde a}_{t^{'}},
\end{align}
where ${{\bf{\Lambda }}_2} = {{\tilde \lambda }_{\max }}({{\bf{R}}_2}){{\bf{I}}_N}$, with ${{\tilde \lambda }_{\max }}({{\bf{R}}_2})$ denoting the maximum eigenvalue of ${{\bf{R}}_2}$.
Then, the objective function is expressed in terms of  $\tilde{\mathbf{a}}$ as
$ {{\bar f}_4}({\bf{\tilde A}}) = \Re \left\{ {{{{\bf{\tilde r}}}_3}^H{\bf{\tilde a}}} \right\}$,
where ${{{\bf{\tilde r}}}_3} = 2{{\bf{\tilde{r}}}_1} + 2{{\bf{\tilde{r}}}_2} + 2({{\bf{R}}_2} - {{\bf{\Lambda }}_2})\mathbf{\tilde a} _{t^{'}}$.
Thus, the sub-problem w.r.t. $\tilde{\mathbf{a}}$ is formulated as
\begin{subequations}\label{pro: a tilde}
 \begin{align}
 \mathop {\min }\limits_{\tilde{\mathbf{a}}} \quad
 & \Re \left\{ {\bf{\tilde r}_3}^H{\tilde{\mathbf{a}}} \right\}
 \\
 \quad\  \textrm{s.t.}\quad
 & \sum\nolimits_{i = 1}^N \bar{{a}}_{l,i} = 1, \; l=1,\cdots, a,\label{con: sum a = 1} \\
 & m\neq n, \; m,n \in \{v|\bar{{a}}_{l,v} = 1, l = 1,\cdots,a\},\label{con: different index in each N elements}
 \end{align}
\end{subequations}
where $\bar{{a}}_{l,i}$ denotes the $i$-th entry of the $l$-th segment in $\tilde{\bf{a}}$. Note that each segment of $\tilde{\bf{a}}$ consists of $N$ elements, which come from the corresponding column of $\tilde{\bf{A}}$. For problem \eqref{pro: a tilde}, constraint \eqref{con: sum a = 1} is the constraint of the total number of selected elements in each column of $\tilde {\bf{A}}$. Constraint \eqref{con: different index in each N elements} guarantees that $a$ different elements should be selected to operate in connection mode.

By temporarily relaxing constraint \eqref{con: different index in each N elements},
the solution to problem \eqref{pro: a tilde} is given by
\begin{equation}\label{equ: opt_tilde a}
{\bar a}_{l,i}\!\! =
\begin{cases}
 1, & \!\!i\!=\!\mathop{\arg\min}\limits_{m}\{{\bf c}[(l-1)N+1:lN]_{m}\}, \\
 0, & \!\!\textrm{Otherwise}, 
\end{cases}
\end{equation}
where ${\bf c} = \Re \left\{ {\bf{\tilde r}_3}\right\}$, $l\in\{1,\cdots,a\} $, $i,m\in\{1,\cdots, N\}$, and $\mathop{\arg\min}\limits_{m}\{{\bf c}[(l-1)N+1:lN]_{m}\}$ returns the index corresponding to the minimum value of ${\bf c}[(l-1)N+1:lN]$.
It is observed from \eqref{equ: opt_tilde a} that ${\bar a}^{\rm{\mathrm{opt}}}_{l,i} = 1$ holds when the $i$-th entry of ${\bf c}[(l-1)N+1:lN]$ is the minimum value among the $N$ elements of the l-th block of ${\bf c}$. By considering the constraint \eqref{con: different index in each N elements}, the mode switching matrix optimization should be further discussed in the following. 

Let $\mathcal{T} = \{i|\bar{a}^{\rm{opt}}_{l,i}=1,l=1,\cdots,a\}$ denote the set of selected indices and $M$ represent the number of distinct index types in $\mathcal{T}$. If $M = a$, the solution to problem \eqref{pro: a tilde} is given in \eqref{equ: opt_tilde a}.
If $M < a$, the number of selected connection elements is insufficient, in other words, constraint \eqref{con: different index in each N elements} is not satisfied. In this case, we need to reselect the index for the segments with repeated indexes. Specifically, We assign the repeated index to the segment that minimizes the objective function of problem \eqref{pro: a tilde} and then select the index corresponding to the sub-minimum value in the remaining segments until $M = a$. 

The above repeated reselection steps ensure that both constraint \eqref{con: different index in each N elements} holds and the objective function of problem \eqref{pro: a tilde} is minimized. Moreover, in the $(t+1)$-th round, the penalty term $\rho$ in the objective function is updated in a moderate step size $\eta \in (0,1)$, which is given by
\begin{equation} \label{equ: update rho}
\rho^{(t+1)} = \eta \rho^{(t)}.
\end{equation}

\begin{algorithm}[t]  \small
    \label{Algo: 1}
    \SetAlgoLined %显示end
	\caption{PWM Algorithm for Joint Beamforming and Mode Switching}%算法名字
	\KwIn{$K$, $N_{\rm{t}}$, $N$, $a$, $\mathbf{G}$, ${\bf{h}}_{{\rm{r}},k}$, ${{\sigma _k}}$, $\alpha_k$. }%输入参数
	Randomly initialize $\bm \varphi$, $\bf{a}$, $\bf{\tilde a}$ \;\label{step: 1 in A1} %\;用于换行
        Initialize $\mathbf{F} = [{\bf W}_{\mathrm{b}}^{T}$, ${\bf W}_{\mathrm{r}}^{T}]^{T}$ by the adaptive MRT-ZF initialization\; \label{step: 2 in A1}
        Calculate $u_k$ and ${{\lambda _k}}$ according to  \eqref{equ: opt_lammda} and \eqref{equ: opt_u} \;\label{step: 3 in A1}
        Update ${\bf W}_{\mathrm{b}}$, ${\bf W}_{\mathrm{r}}$ according to \eqref{equ: opt_f_R}\;\label{step: 4 in A1}
        \Repeat{the convergence is satisfied} 
        {
        Calculate $u_k$ and ${{\lambda _k}}$ according to \eqref{equ: opt_u} and \eqref{equ: opt_lammda};\label{step: 6 in A1}\\ 
        Update $\bm {\varphi}$ according to \eqref{equ: power iteration} until the objective function of \eqref{pro: rankone} converges;\label{step: 7 in A1}\\
        Update $\bf \tilde{A}$ according to \eqref{equ: opt_tilde a};\label{step: 8 in A1}\\
        Update $\bf A$ according to \eqref{equ: opt_a};\label{step: 9 in A1}\\
        Update $\mathbf{F}$ according to \eqref{equ: opt_f_R} and scale $\mathbf{F}$ according to \eqref{con: P};\\
        Update $\rho$ according to \eqref{equ: update rho}; \label{step: 11 in A1}
        }
	Return ${\bf A}^{\mathrm{opt}} = \operatorname{diag}(\bf{a}^{\mathrm{opt}})$, ${\bf{\tilde A}}^{\mathrm{opt}}$, ${\mathbf{W}}_{\rm{b}}^{\mathrm{opt}}$, ${\bf W}_{\rm{r}}^{\mathrm{opt}}$, ${\bf{\Phi}}^{\mathrm{opt}} = \operatorname{diag}(\bm{\varphi}^{\mathrm{opt}})^{H}$\;
    \KwOut{${\bf W}_{\mathrm{b}}$, ${\bf W}_{\mathrm{r}}$, ${\bf \Phi}$, ${\bf A}$, ${\bf {\tilde A}}$.}%输出
\end{algorithm}

The proposed PWM algorithm for solving problem \eqref{pro: weighted sum rate} is summarized in Algorithm \ref{Algo: 1}. The passive beamforming and index matrices are randomly initialized in Step \ref{step: 1 in A1}. The equivalent active beamforming matrix is initialized via the adaptive maximum ratio transmission (MRT) and zero-forcing (ZF) methods in Step \ref{step: 2 in A1}. Specifically, we apply the MRT method to enhance the strength of the transmit signal in the low-SNR regime. While for the high-SNR regime, the ZF method is utilized to suppress the inter-user interference (IUI). With the auxiliary variables calculated in Step \ref{step: 3 in A1}, we update the active begriming matrices in Step \ref{step: 4 in A1}. Finally, the mode switching matrix is updated based on other variables according to the MM method from Steps \ref{step: 6 in A1} to \ref{step: 11 in A1}. 

\subsection{Convergence and Complexity Analysis}
The sub-problem for updating ${\bm{\varphi}}$ is optimally solved and thus the objective function of problem \eqref{pro: rankone} is maximized, which is equivalent to problem \eqref{pro: weighted sum rate MMSE}. Moreover, the objective value of problem \eqref{pro: A} is non-increasing via the MM method, which guarantees the convergence of Algorithm \ref{Algo: 1}. Moreover, the computational complexity of the proposed algorithm is analyzed as follows. Specifically, for the active beamforming optimization, the computational complexity depends on the number of UEs $K$ and the numbers of BS antennas and RDARS transmit elements, given by  $\mathcal{O}(K(N_{\rm{t}} + a)^3)$. For the passive beamforming optimization, the calculation of $\bf{C}$ dominates the complexity, which is given by $\mathcal{O}(K^2 N^2)$. In addition, problem \eqref{pro: rankone} is solved by power iteration, and the complexity is $\mathcal{O}(I_{p}N^{2})$ where $I_{p}$ denotes the inner iteration steps required for convergence. For the mode switching matrix, the complexities of optimizing $\bf{A}$ and $\tilde{\bf{A}}$ mainly arise from eigenvalue decomposition (EVD), which are given by $\mathcal{O}((2N)^3)$ and $\mathcal{O}(N^3)$, respectively. Therefore, the total complexity of Algorithm $\ref{Algo: 1}$ is $\mathcal{O}(I(K(N_{\rm{t}} + a)^3 + K^2 N^2 + I_{p}(N)^{2} + 3 N^2) )$, where $I$ denotes the number of iterations.

\section{Model-Driven PWM-BFNet }
In this section, we propose a model-driven PWM-BFNet to mitigate the number of iterations required by PWM algorithm and escape from local optima, thus reducing the computational complexity while improving performance. 

Specifically, the PWM-BFNet leverages DL to learn crucial parameters associated with the convergence speed, i.e., the penalty term $\rho^{(t)}$ in each iteration and regularization parameter $\varepsilon^{(t)}$. 
After an in-depth presentation of the model-driven DL, a comprehensive comparison and analysis of the complexity for the proposed techniques and other existing algorithms is provided.

\subsection{Convergence Acceleration}
Though the convergence of the proposed PWM algorithm can be guaranteed, its convergence speed depends on the penalty term $\rho$ and regularization term $\varepsilon$.

% state the necessity of update \rho
Generally, it is preferable to set $\rho$ to moderate values during initialization, ensuring that the penalized objective function is primarily dominated by the original objective function rather than the penalty terms \cite{zhang_RDARS, wu2020joint}.
Specifically, it is observed from problems \eqref{pro: A}, \eqref{pro: tilde A} and Algorithm \ref{Algo: 1} that $\rho$ should be sufficiently large in the initial steps, thus ensuring that the first terms for the objective function of problem \eqref{pro: A} and \eqref{pro: tilde A} dominate these function values. Meanwhile, the aims of problem \eqref{pro: A} and \eqref{pro: tilde A} are to minimize the corresponding first terms.
When the beamforming matrices render the values of the first terms tending to be stable, a suitable $\rho$ guarantees that the associated terms dominate these function values, and the iteration processes aim to decrease the corresponding terms.
Therefore, when the penalty term of each iterative step is not suitable, the convergence speed of problem \eqref{pro: A} and \eqref{pro: tilde A} is greatly limited to by considerable degree, thus influencing the performance of the proposed algorithms. 
To mitigate the impact of the penalty term on the outer iteration, a model-driven DL is designed by selecting $\rho$ as a trainable variable.

In addition, another inner iteration related to $\varepsilon$ in \eqref{equ: power iteration} impacts the convergence speed of the proposed algorithm. Specifically, when $\varepsilon \rightarrow +\infty$, ${\bf{p}}^{(q+1)}= e^{j\mathrm{arg}(({\bf{D}}+\varepsilon {\bf{I}}_{N+1}){\bf{p}}^{(q)})} \rightarrow e^{j\mathrm{arg}(\varepsilon{\bf{p}}^{(q)})}={\bf{p}}^{(q)}$, thus causing the iteration meaningless. On the other hand, when $\varepsilon$ is too small, the positive definiteness of ${\bf{D}}+{\bf{I}}_{N+1}$ cannot be guaranteed, which may compromise the convergence speed.
It is observed from Algorithm \ref{Algo: 1} that the passive beamforming is updated by ensuring that the objective function of problem \eqref{pro: passive BF} converges in step \ref{step: 7 in A1}, and the complexity of this step is linear with the number of inner iterations. Therefore, we introduce a model-driven DL to train $\varepsilon$, and set the number of inner iterations as 1, so as to reduce the complexity of updating passive beamforming and accelerate the convergence speed of the proposed algorithm.

Moreover, the common data-driven method depends on the suitable neural network to train beam indices or complex beamforming matrices, where a high-quality dataset is needed. Besides, the large number of hidden layers in training network leads to a large amount of trainable parameters, which increases the training complexity and memory requirements. Therefore, by integrating the expert model
knowledge, we propose a model-driven DL to train a small number of parameters. To be specific, the optimal active beamforming vector in \eqref{equ: opt_f_R} can be represented by a simple solution structure as follows\cite{simple_solutin_structure_Emil},
\begin{equation}\label{equ: simple structure of opt_f_R}
 \mathbf{f}^{\rm{opt}}_{k} = \sqrt{p_{k}}\frac{(\mathbf{I}_{N_{\rm{t}}+a}+\sum_{m=1}^{K}\frac{\delta_{m}}{\sigma^2}\mathbf{h}^{H}_{m}\mathbf{h}_{m})^{-1}\mathbf{h}^{H}_{k} }{|| (\mathbf{I}_{N_{\rm{t}}+a}+\sum_{m=1}^{K}\frac{\delta_{m}}{\sigma^2}\mathbf{h}^{H}_{m}\mathbf{h}_{m})^{-1}\mathbf{h}^{H}_{k}||},
\end{equation}
with 
\begin{equation}\label{con: p and delta for UE}
\sum_{m=1}^{K}{p_m} = \sum_{m=1}^{K}{\delta_{m}} = P_{\rm{tot}}.
\end{equation}
To achieve a faster convergence speed, the active beamforming vector $\mathbf{f}_{k}$ can be initialized based on the simple structure in \eqref{equ: simple structure of opt_f_R}, and the parameters $p_k$ and $\delta_k$ are trainable parameters in the deep unfolding network. It is observed from the constraint \eqref{con: p and delta for UE} that the total value for $p_k$ and $\delta_k$ should be achieved. Therefore, the trainable parameters $p_k$ and $\delta_k$ are scaled by the total power constraint through a softmax layer as follows.
\begin{equation}
  p_k = \frac{e^{p^{'}_{k}}}{\sum_{m=1}^{K} e^{p^{'}_{m}}} P_{\rm{tot}},
\end{equation}
\begin{equation}
\delta_k = \frac{e^{\delta^{'}_{k}}}{\sum_{m=1}^{K} e^{\delta^{'}_{m}}} P_{\rm{tot}},
\end{equation}
where $p^{'}_{m}$ and $\delta^{'}_{m}$ are set as the input of the softmax layer. 

For the network training, the unsupervised learning is adopted to train the PWM-BFNet. Since the goal of problem \eqref{pro: weighted sum rate} is to maximize WSR, the loss function is formulated as:
\begin{equation}
\mathcal{L} = - \frac{1}{N_{\rm{s}}}\sum_{n=1}^{N_{\rm{s}}} f\left(\textrm{PWM-BFNet}(\mathbf{G}^{(n)}, \widetilde{\mathbf{H}}^{(n)}_{\rm{r}}) \right),
\end{equation}
where $N_{\rm{s}}$ is the number of channel realizations in the training dataset, $f(\cdot)$ denotes the objective function of problem \eqref{pro: weighted sum rate}, and $\textrm{PWM-BFNet}(\mathbf{G}^{(n)}, \widetilde{\mathbf{H}}^{(n)}_{\rm{r}})$ is the beamformer obtained by PWM-BFNet taking the $n$-th channel realization $\{\mathbf{G}^{(n)}, \widetilde{\mathbf{H}}^{(n)}_{\rm{r}}\}$ as input. Let $ \widetilde{\mathbf{H}}_{\rm{r}} = [{\bf{h}}_{{\rm{r}},1}, {\bf{h}}_{{\rm{r}},2}, \cdots,  {\bf{h}}_{{\rm{r}},K}]$.
The proposed model-driven PWM-BFNet algorithm is summarized in Algorithm \ref{Algo: 4}.
\begin{figure*}[t] 
 \centering
 \includegraphics[width=0.9\textwidth]{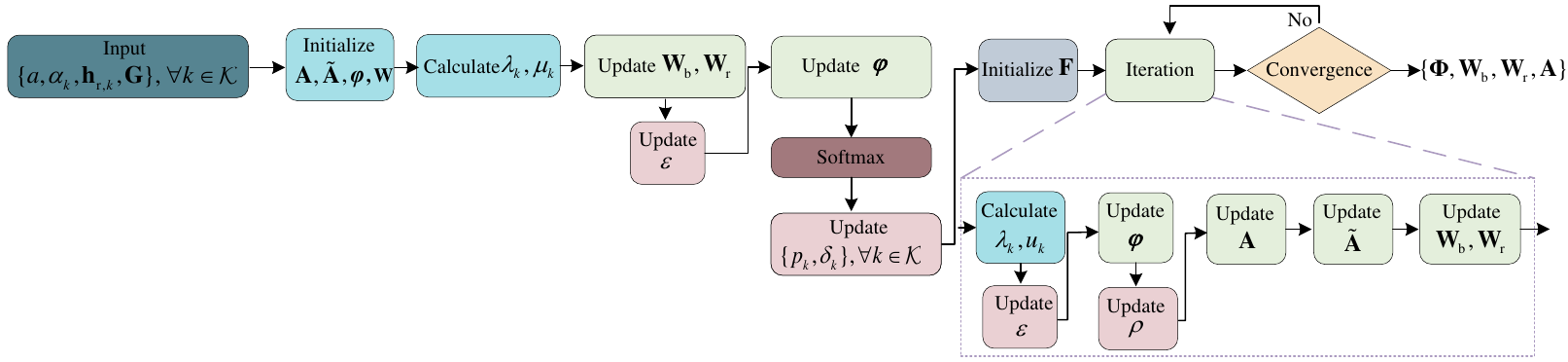}
 \caption{An illustration of the network structure for PWM-BFNet.}
 \label{fig: PWM-BFNet+}
\vspace{-5pt}
\end{figure*}
\begin{algorithm}[t]\small
    \label{Algo: 4}
    \SetAlgoLined %显示end
	\caption{PWM-BFNet  Algorithm for Joint Beamforming and Mode Switching}%算法名字
	\KwIn{$K$, $N_{\rm{t}}$, $N$, $a$, $\mathbf{G}$, ${\bf{h}}_{{\rm{r}},k}$, ${{\sigma _k}}$, $\alpha_k$. }%输入参数
	Randomly initialize $\bf{a}$, $\bf{\tilde a}$ and $\bm \varphi$\;\label{step: 1 in A4} %\;用于换行
        Initialize $\mathbf{F} = [{\bf W}_{\mathrm{b}}^{T}$${\bf W}_{\mathrm{r}}^{T}]^{T}$ with ZF beamforming\;  \label{step: 2 in A4}
        Calculate ${{\lambda _k}}$ and $u_k$ according to \eqref{equ: opt_u} and \eqref{equ: opt_lammda}\;\label{step: 3 in A4}
        Update ${\bf W}_{\mathrm{b}}$, ${\bf W}_{\mathrm{r}}$ according to \eqref{equ: opt_f_R}\;\label{step: 4 in A4}
        Update $\bm \varphi$ according to \eqref{equ: power iteration} and trainable variable $\varepsilon^{(0)}$\;\label{step: 5 in A4}
        Initialize $\mathbf{F}$ with $p_{k}$ and $\delta_k$, $\forall k \in \mathcal{K}$ \;
        \Repeat{the convergence is satisfied.} 
        {
        Calculate ${{\lambda _k}}$ and $u_k$ according to \eqref{equ: opt_u} and \eqref{equ: opt_lammda};\label{step: 6 in A4}\\ 
        Update $\bm {\varphi}$ according to \eqref{equ: power iteration} and trainable variable $\varepsilon^{(i)}$;\label{step: 7 in A4}\\
        Update $\bf \tilde{A}$ according to \eqref{equ: opt_tilde a} and trainable variable $\rho^{(i)}$;\label{step: 8 in A4}\\
        Update $\bf A$ according to \eqref{equ: opt_a} and trainable variable $\rho^{(i)}$;\label{step: 9 in A4} \\
        Update $\mathbf{F}$ according to \eqref{equ: opt_f_R} and scale $\mathbf{F}$ according to \eqref{con: P};\\
        \label{step: 11 in A4}
        }
	Return ${\bf A}^{\mathrm{opt}} = \operatorname{diag}(\bf{a}^{\mathrm{opt}})$, ${\bf{\tilde A}}^{\mathrm{opt}}$, ${\mathbf{W}}_{\rm{b}}^{\mathrm{opt}}$, ${\bf W}_{\rm{r}}^{\mathrm{opt}}$, ${\bf{\Phi}}^{\mathrm{opt}} = \operatorname{diag}(\bm{\varphi}^{\mathrm{opt}})^{H}$\;
    \KwOut{${\bf W}_{\mathrm{b}}$, ${\bf W}_{\mathrm{r}}$, ${\bf \Phi}$, ${\bf A}$, ${\bf {\tilde A}}$.}%输出
\end{algorithm}
\subsection{Complexity Analysis}
Next, we compare the complexity of the proposed algorithm with other joint beamforming algorithms, including the PWM, PWM-BFNet, RCB-BF \cite{ji2025reconfigurable}, and DLRCB-PWM algorithms, which are listed in Table \ref{tab 1: complexity comprasion}. For the RCB-BF algorithm, the computational complexity of active beamforming optimization arises from the codeword selection, i.e., $\mathcal{O}(K(N^2_{\rm{t}}+a^2))$. For the passive beamforming optimization, the water-filling iteration algorithm dominates the complexity, which is $\mathcal{O}(I_{\rm{o}}(KN^3 + K^2 N (N+a)))$, where $I_{\rm{o}}$ denotes the number of outer iterations. Thus, the total complexity of the RCB-BF algorithm is $\mathcal{O}(I_{\rm{o}}(KN^3 + K^2 N (N+a)) + K(N^2_{\rm{t}}+a^2))$.  
For the proposed DLRCB-PWM algorithm, the complexity of each iteration is the same as the PWM algorithm, where the initializations of beamforming matrices are obtained according to the RCB-based beamforming design. Due to the DLRCB initialization, the number of inner and outer iterations can be reduced, which leads to a lower complexity. Therefore, the total complexity of the DLRCB-PWM algorithm is $\mathcal{O}(I_{\rm{d}}(K(N_{\rm{t}} + a)^3 + K^2 N^2 + I_{\rm{pd}}N^{2} + 3 N^2) )$, where $I_{\rm{d}}$ denotes the number of outer iterations and $I_{\rm{pd}}$ denotes the number of inner iterations.
For the PWM-BFNet algorithm, the number of inner iteration steps is reduced to 1, which significantly decreases the complexity. Furthermore, the computational complexity of active beamforming is only caused by the initial step that reconstructs the beamforming vectors according to \eqref{equ: simple structure of opt_f_R}, i.e., $\mathcal{O}(K(N_{\rm{t}} + a)^2)$. Thus, the total complexity of the PWM-BFNet algorithm is 
$\mathcal{O}(I^{'}_{\rm{d}}(K(N_{\rm{t}} + a)^2 + K^2 N^2 + 5 N^3) )$, where $I^{'}_{\rm{d}}$ denotes the number of outer iterations. It is observed that the active beamforming vectors initialized by the PWM-BFNet algorithm are closer to the optimal solution, which indicates $I^{'}_{\rm{d}} < I_{\rm{d}} \ll I$, thus leading to a lower complexity and faster convergence speed.

\begin{table}
\caption{\textbf{Complexity Analysis}}%标题
\centering%把表居中
\begin{threeparttable}
\begin{tabular}{cccc}%四个c代表该表一共四列，内容全部居中
\toprule%第一道横线
\textbf{Algorithm}&\textbf{Computational Complexity} \\
\midrule%第二道横线 
PWM&$\mathcal{O}(I(K(N_{\rm{t}} + a)^3 + K^2 N^2 + I_{\rm{p}}N^{2} + 5 N^3) )$\\
RCB-BF&$\mathcal{O}(I_{\rm{o}}(KN^3+K^2N(N+a))+ K(N^2_{\rm{t}}+a^2))$ \\
DLRCB-PWM&$\mathcal{O}(I_{\rm{d}}(K(N_{\rm{t}} + a)^3 + K^2 N^2 + I_{\rm{pd}}N^{2} + 5 N^3) )$ \\
PWM-BFNet &$\mathcal{O}(I^{'}_{\rm{d}}(K(N_{\rm{t}} + a)^2 + K^2 N^2 + 5 N^3) )$ \\
\bottomrule%第三道横线
\end{tabular}
\label{tab 1: complexity comprasion}
\end{threeparttable}
\end{table}

% \end{figure*}

\section{Simulation Results}
In this section, numerical results are provided to verify the proposed schemes. The number of UEs is $K=4$, the number of antennas at BS is $16$, and the number of total elements for RDARS is $128$, respectively. The number of TEs is $a = 8$. 
The BS and the RDARS are located at (0, 0, 15) m and (10, 0, 15) m, respectively. The UEs are randomly distributed within a circle, where its center and radius are (10, 50, 2) m and 5 m, respectively. The path loss models are given by $c_0(\frac{d}{D_0})^{-\delta}$, where $c_0$ is the path loss at the reference distance $ D_0 =1$ m, $d$ denotes the link distance, and $\delta$ denotes the path loss exponent. We set $c_0 = 60.4$ dB. The path loss exponents of the BS-RDARS and RDARS-UE $K$ links are $\delta_{\mathrm{b}} = 2.2$ and $\delta_{\mathrm{r},k} = 2.4$, respectively. Other system parameters are set as follows: $\xi=10$, $\rho^{(0)} = 10^{6}$, $\eta = 10^{-3}$, and $\sigma_{k}^2 = -80$ dBm.
\iffalse
The DLRCB initialization network takes the channel $\tilde{\mathbf{H}}_{{\rm{r}}}$ as input and processes it through an FC layer with $D=128$ neurons. This is followed by a $4$-layer transformer encoder with a hidden dimension of $d=256$ and $h=8$ self-attention heads. Finally, the output $\hat{\mathbf{H}}_{{\rm{r}}}$ is passed through MLP networks, each consisting of two FC layers, which include a hidden layer with $64$ neurons, a Rectified Linear Unit (ReLU) activation function, and an output layer matching the corresponding beam codebook dimension. The whole network will be trained for $100$ epochs, and the dataset contains $100000$ channel realizations, where the AoA and AoD for each UE are chosen from $N$ non-overlapping equally-spaced subregions \cite{Elbir2020}. Specifically, each subregion can be regarded as a beam coverage, given by $\Theta^{\rm{u}}_{i} = [-1 + \frac{2i-2}{N}, -1+\frac{2i}{N}], i\in \mathcal{N}$, such that $\cup_{i\in \mathcal{N}} \Theta^{\rm{u}}_{i} =\Theta^{\rm{u}} =[-1,1]$. Similarly, the AOD between the BS and RDARS is chosen from $\Theta^{\rm{b}}$, given by $\cup_{i\in \mathcal{N_{\rm{t}}}} \Theta^{\rm{b}}_{i} =\Theta^{\rm{b}} =[-1,1]$ with $\Theta^{\rm{b}}_{i} = [-1 + \frac{2i-2}{N_{\rm{t}}}, -1+\frac{2i}{N_{\rm{t}}}]$. Furthermore, the RCBs corresponding to the BS, RDARS transmit elements, and passive elements are generated to obtain the optimal beam index, where the resolution of passive RCB is non-oversampling to reduce the complexity of beam classification.
\fi
For the PWM-BFNet algorithm, we set $\{\tilde{\rho}_{\rm{ini}}, \tilde{\rho}^{(i)}, \tilde{\varepsilon}_{\rm{ini}}, \tilde{\varepsilon}^{(i)}, p_k, \delta_k\}$ as trainable parameters, and the corresponding parameters are substituting into the PWM algorithm for $1\le i \le I^{'}_{\rm{d}}$ and $k \in \mathcal{K}$, where $\tilde{\rho}_{\rm{ini}} = \rho^{(0)}$, $\tilde{\rho}^{(i)} = \rho^{(i)}$, and $\tilde{\varepsilon}_{\rm{ini}} = \tilde{\varepsilon}^{(0)}=\varepsilon^{(0)}$ at the initialization step. Similarly, we generate 10,000 channel realizations to train the model-driven network. Note that the label corresponding to each sample is not required due to the unsupervised learning. PyTorch, a DL-based architecture, is used for building and training the proposed algorithms, and the networks are trained by using the SGD optimizer. The momentum is set as 0.7. The number of batches in each epoch is $1000$, where the number of samples in each batch is 10. Moreover, the number of training epochs is 30. The main simulation parameters are summarized in Table \ref{tab: 2}.
 \begin{table}[t]
\centering % no center environment
\begin{threeparttable}
\caption{The Main Simulation Parameters.}
\label{tb:listbridges}
\begin{tabular}{llll}
\toprule
\textbf{Description} & \textbf{Parameter} & \textbf{Value} \\
\midrule
Number of transmit antennas at BS & $N_{\rm{t}}$ & 16 \\
Number of RDARS elements& $N$ & 128 \\
Number of connected elements& $a$ & 8\\
Noise power& $\sigma_{k}^2$ & -80 dBm \\
Number of UEs& $K$ & 4 \\
Path loss exponent of channel ${\bf{G}}$& $\delta_b$ & 2.2 \\
Path loss exponent of channel ${\bf{h}}_{{\rm{r}},k}$& $\delta_{{\rm{r}},k}$ & 2.4 \\
Penalty term in the initialization step& $\rho^{(0)}$ & $10^{6}$ \\
Step size& $\eta$ & $10^{-3}$ \\
Number of batches&  & 1,000 \\
Number of samples&  & 10 \\
Number of training epochs&  & 30 \\
Number of channel realizations&  & 10,000\\
Momentum&  & 0.7\\
Optimizer& &SGD\\
\bottomrule
\end{tabular}
\label{tab: 2}
\end{threeparttable}
\end{table}
%% 补充上learning rate, decay, momentum, optimizer等参数

For comparison, the following architectures are considered:
\textbf{1) RDARS PWM-BFNet:}  The RDARS-aided system with the PWM-BFNet is considered; \textbf{2) RDARS PWM:} The RDARS-aided system with the PWM algorithm is considered; \textbf{3) RDARS DLRCB-PWM:} The RDARS-aided system with the DLRCB-PWM algorithm is considered, where the initializations of beamforming matrices are generated by RCBs; \textbf{4) RDARS Fixed Index:} The RDARS-aided system with the PWM algorithm is considered, where the the working modes of RDARS elements are fixed; The first $a$ elements are chosen as the transmit elements, and the others are passive elements; \textbf{5) DAS:} The DAS is considered, where the number of BS antennas $\tilde{N}_{\rm{t}}$ is equal to $N_{\rm{t}}$, and the number of distributed antennas $\tilde{a}$ is equal to $a$. The active beamforming matrix is optimized based on the PWM algorithm, which indicates that $a$ distributed antennas are located in optimal positions determined in the RDARS-aided system; \textbf{6) RIS:} The RIS-aided system is considered, where the number of RIS elements is $\tilde{N} = N$. The active and passive beamforming matrices are optimized based on the PWM algorithm.
Furthermore, the performance of the PWM-BFNet trained in different parameters is compared. 

\subsection{Convergence Analysis}
\begin{figure}[t] 
 \centering
  \includegraphics[width=0.4\textwidth]{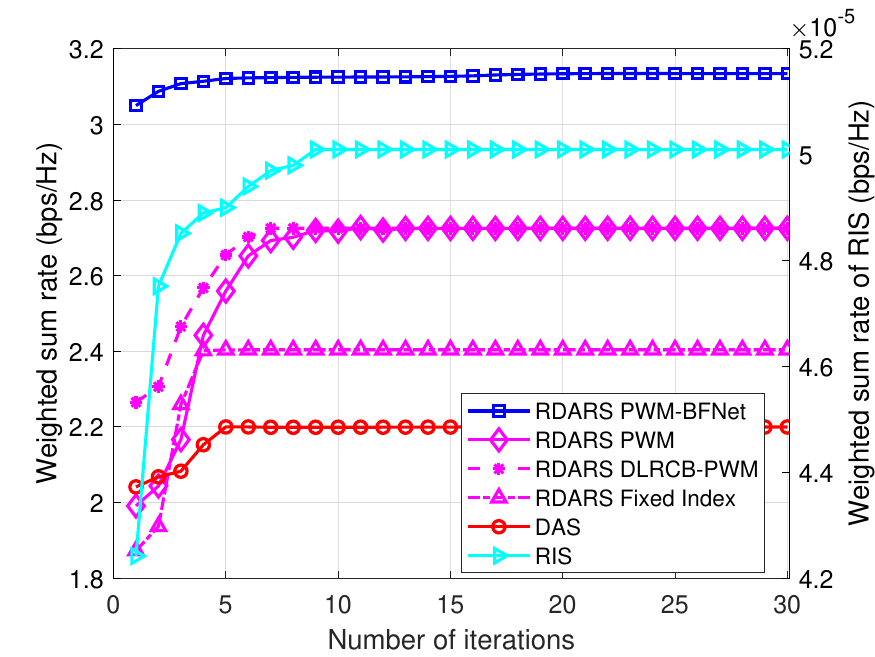}
   \vspace{-5pt}
 \caption{Convergence behaviors of proposed algorithms. The curves corresponding to RIS are scaled by the right y-axis, while the curves corresponding to the other schemes are scaled by the left y-axis.}
 \label{fig: 1_Number_of_iterations.pdf}
 \vspace{-5pt}
\end{figure}

In Fig, \ref{fig: 1_Number_of_iterations.pdf}, we plot the WSR versus the number of iterations with the different architectures. 
It is observed that the PWM-based algorithm achieves convergence for different architectures, which verifies the effectiveness of this algorithm.  
Besides, the DLRCB-PWM algorithm has a better convergence performance than the PWM algorithm. This is because the RCB-based initialization generates better initialization points for the beamforming matrices. 
It is also observed that the PWM-BFNet significantly improves the performance due to the integration of the model-driven DL. Specifically, the parameters in each iteration step are trained for multiple channel realizations, so as to escape the local optimum of the PWM algorithm. 
Moreover, the performance of the PWM-BFNet algorithm after one iteration attains a significant performance enhancement compared to the PWM algorithm after convergence, and the number of iterations required for convergence is reduced.
By incorporating the simple solution structure of active beamforming and model-driven DL, PWM-BFNet achieves a fast convergence performance, which converges around 5 iterations. Moreover, it is observed that PWM-BFNet yields the best results, which is attributed to its more effective initialization. Specifically, the simple structure solutions of active beamforming matrices reconstructed by trained parameters are close to spatial directions of the optimal solutions. In addition, the inner iterations of the passive beamforming are reduced to 1, which accelerates the convergence.

Furthermore, the RDARS architecture achieves a better performance than the DAS and RIS architectures, even for the fixed element configuration, which demonstrates the inherent benefits of RDARS. Moreover, the WSR of the RDARS-aided system with mode switching optimization is significantly improved compared with the fixed index scheme, which verifies the necessity of optimizing the mode switching.

%\vspace{-5pt}
\subsection{Robustness Analysis}
In this subsection, the performance of the proposed algorithms is evaluated by considering the impacts of total transmit power $P_{\rm{tot}}$, numbers of UEs $K$, numbers of RDARS elements $N$, and the Rician factor $\xi$. Specifically, the curves of RDARS with the PWM-BFNet algorithm (trained at corresponding parameters) denote that the networks to be tested are trained at corresponding parameters. By contrast, the curves of RDARS with the PWM-BFNet algorithm (trained at a given parameter) indicate that the networks to be tested are trained always at a given parameter. For example, the blue curve using the cross-shaped markers in Fig. \ref{fig: 2_Ptot for different archi} indicates that the test networks are all trained at $P_{\rm{tot}} = 30$ dBm, while the blue square curve represents the test networks are trained at corresponding total transmit powers.

It is observed from Figs. \ref{fig: 2_Ptot for different archi}-\ref{fig: 7_number of TEs} that the proposed PWM-BFNet outperforms other algorithms in various scenarios, by leveraging the inherent advantages of the model-driven DL and the proposed PWM algorithm. 
It is worth mentioning that the proposed PWM-BFNet maintains a superior performance under different parameters which implies its robustness to variations in system setups. 
This is attributed to the ability of PWM-BFNet in capturing the mapping function from channels to beamforming vectors and model switching matrix, a task that is arguably performed more effectively than conventional data-driven DL approaches. 
This inherently improves performance and robustness. Furthermore, the slight performance degradation of the DL method due to differences between training and test environments is mitigated by subsequent iterations.

\subsubsection{Impact of The Total Transmit Power}
\begin{figure}[t] 
 \centering
  \includegraphics[width=0.4\textwidth]{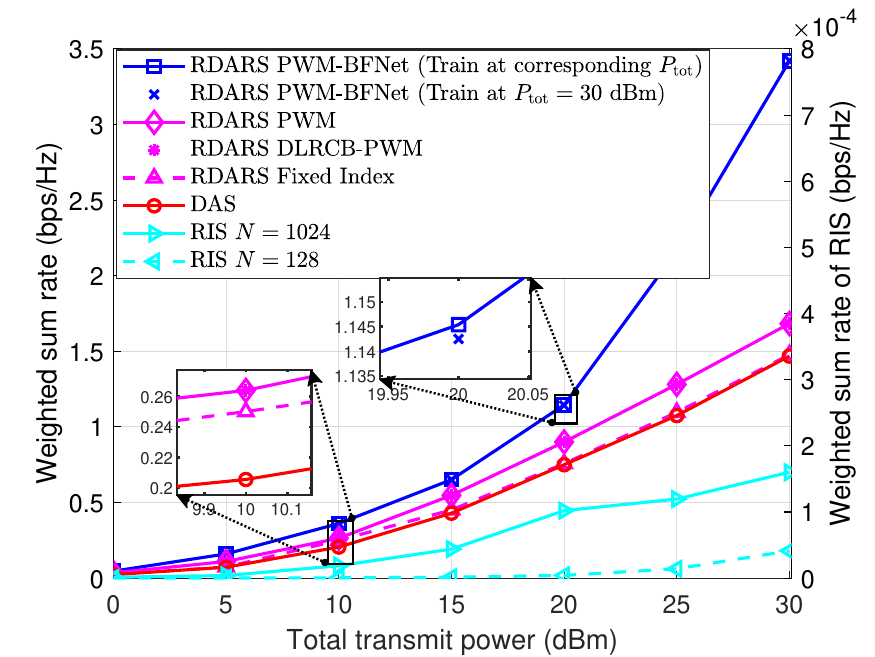}
\caption{WSR versus total transmit power with different architectures.}
\label{fig: 2_Ptot for different archi}
\end{figure}
Fig. \ref{fig: 2_Ptot for different archi} plots the WSR versus the total transmit power for different architectures. The PWM-BFNet achieves a better performance than the PWM algorithm, and the performance gap increases with the total transmit power. This is because the IUI increases, and the PWM algorithm is prone to fall into local optimality. Furthermore, the RDARS architecture with the fixed index scheme achieves a performance comparable to that of DAS, and the RDARS with the mode switching has a further performance improvement. This demonstrates the extra mode selection gain provided by dynamic element configurations of RDARS.

\subsubsection{Impact of The Number of UEs}
\begin{figure}[t] 
 \centering
  \includegraphics[width=0.4\textwidth]{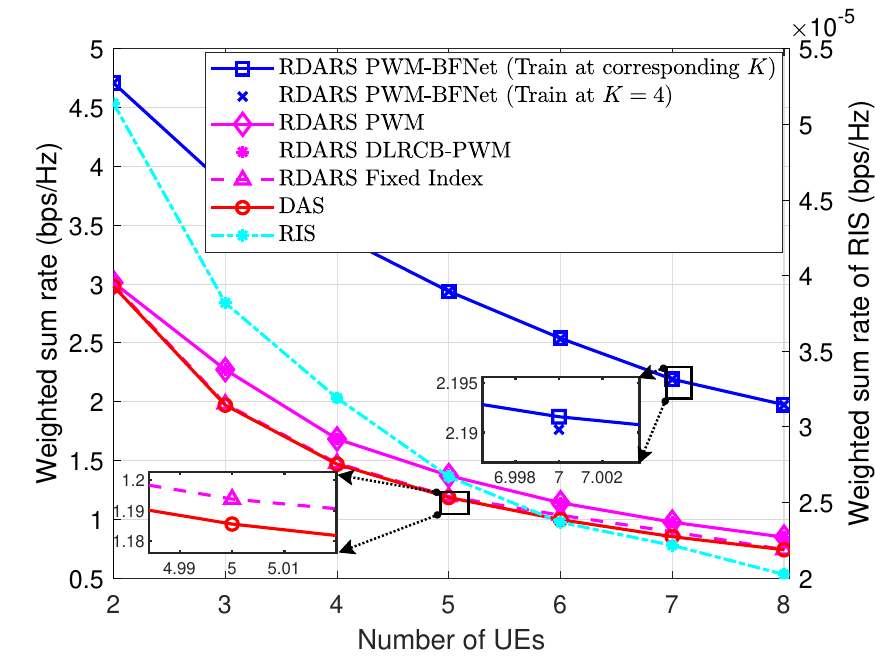}
\caption{WSR versus the number of UEs with different architectures.}
\label{fig: 3_number of UEs}
\end{figure}

In Fig. \ref{fig: 3_number of UEs}, we plot the WSR versus the number of UEs for different architectures.
It is observed that the performance curves of all schemes decrease with the number of UEs. The reason is that the spatial diversity gain is limited as the number of UEs increases, leading to more severe IUI. 
Moreover, the PWM-BFNet trained at $K = 4$ still remains a high performance for scenarios with different numbers of UEs. Specifically, the WSR of PWM-BFNet trained at $K = 4$ when $K = 5$ is comparable to that of the PWM algorithm when $K = 2$, which shows the advantage of model-driven DL. Furthermore, the performance of the RIS-aided system has a notable decline with the number of UEs, which shows that this system is more sensitive to the number of UEs compared with other architectures. This is because the multiplicative fading suffered by RIS significantly impacts the system performance for a large number of UEs.

\subsubsection{Impact of The Number of RDARS elements}
\begin{figure}[t] 
 \centering
  \includegraphics[width=0.4\textwidth]{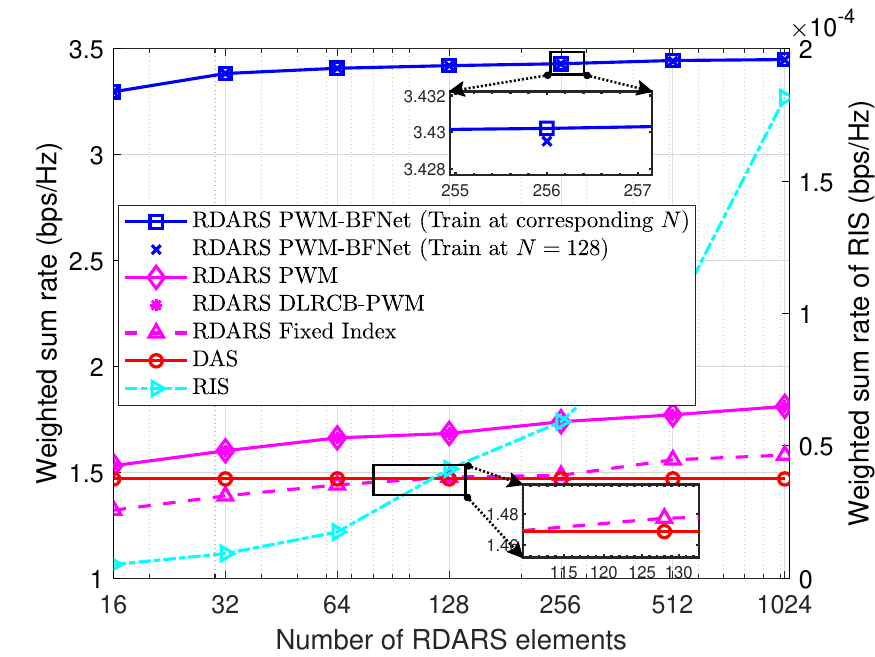}
\caption{WSR versus the number of RDARS elements with different architectures.}
\label{fig: 4_number of RDARS elements}
\vspace{-5pt}
\end{figure}
In Fig. \ref{fig: 4_number of RDARS elements}, we plot the WSR versus the number of RDARS elements with the different architectures.
It is observed that the performance curves of RDARS-aided system increase with the number of RDARS elements. Specifically, the performance enhancement arising from the increase in the number of RDARS elements is limited, which is due to the limited passive beamforming gain provided by passive elements.
Moreover, the RDARS-aided system with around 100 elements using the fixed index scheme achieves a comparable performance to the DAS. This demonstrates the need for dynamic element mode switching, especially for widely fluctuating channel conditions. 

% \begin{figure}[t] 
%  \centering
%   \includegraphics[width=0.4\textwidth]{Conver_PWM_Compare_MRT_ZF_V1030.pdf}
% \caption{Convergence Behavior.}
% \label{fig: PWM iteration}
% \end{figure}

\subsubsection{Impact of The Rician Factor}
\begin{figure}[t] 
 \centering
  \includegraphics[width=0.4\textwidth]{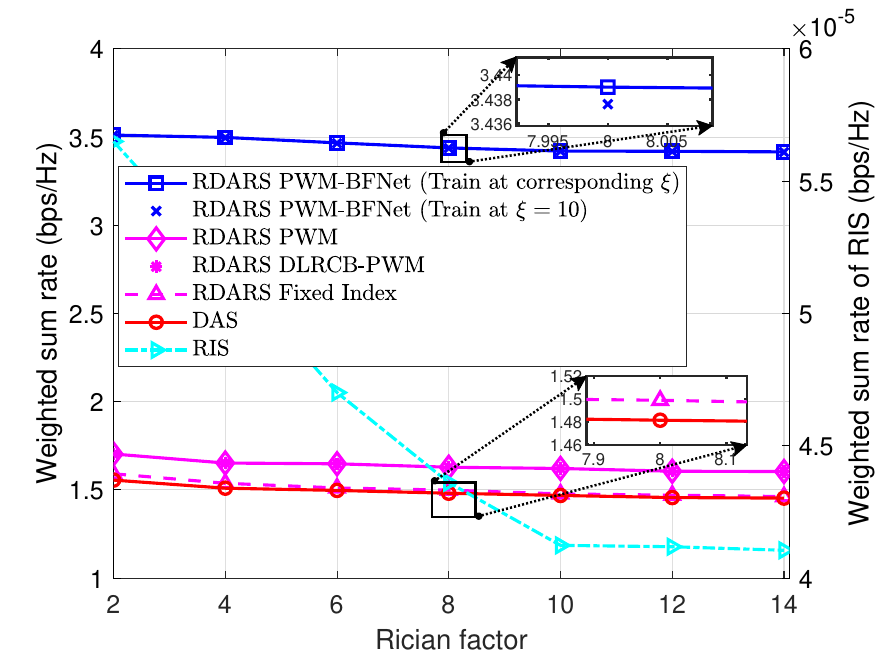}
\caption{WSR versus the Rician factor with different architectures.}
\label{fig: 5_number of rician factor}
\vspace{-5pt}
\end{figure}
In Fig. \ref{fig: 5_number of rician factor}, we plot the WSR versus the Rician factor with the different architectures.
It is seen that the WSR curves decrease as the Rician factor increases. This is because as the Rician factor increases, the line-of-sight (LoS) component dominates the channel, thus resulting in the spatial diversity loss, i.e., the BS-RDARS channel. Thus, the mitigation of IUI is limited due to more significant interference. However, the performance degradation fluctuates within a small range, which further demonstrates the robustness of the proposed algorithms.

\subsubsection{Impact of The Number of TEs}
\begin{figure}[t] 
 \centering
  \includegraphics[width=0.4\textwidth]{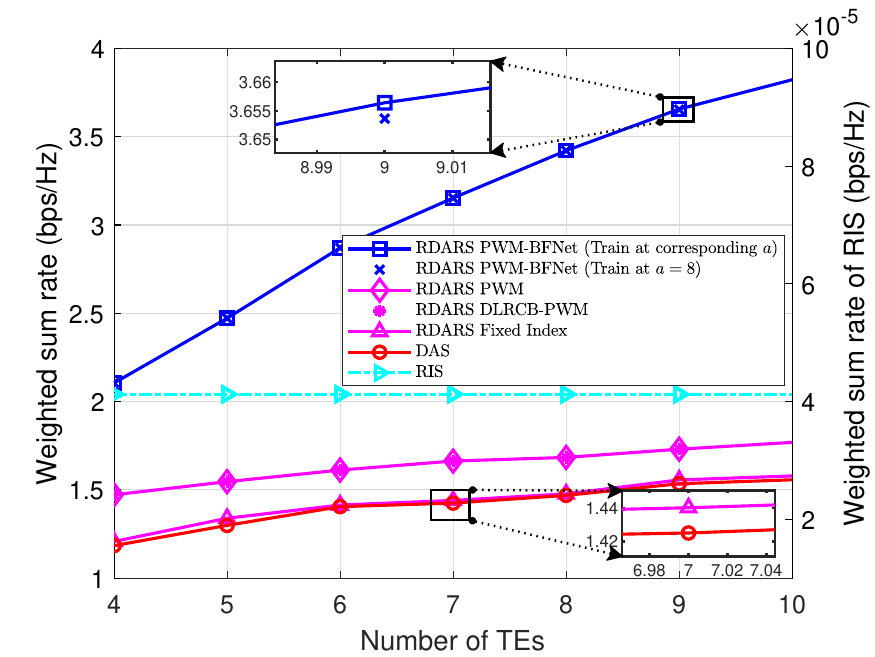}
\caption{WSR versus the number of RDARS transmit elements with different architectures.}
\label{fig: 7_number of TEs}
\vspace{-5pt}
\end{figure}
In Fig. \ref{fig: 7_number of TEs}, we plot the WSR versus the number of RDARS TEs with the different architectures. 
It is observed that the performance of PWM-BFNet has a considerable performance improvement compared to that of the PWM algorithm. The reason is that the PWM algorithm is more easily trapped into a local optimum, due to the intractable integer constraint of the mode switching matrix, especially for a larger number of RDARS TEs.

\section{Conclusion}
In this paper, we considered a RDARS-aided downlink MIMO system, where the WSR was maximized by jointly optimizing the beamforming matrices for the BS and RDARS and the mode switching matrix for RDARS. The PWM algorithm was first proposed to solve the MINLP problem by utilizing the WMMSE, MM, and power iteration algorithms.
\iffalse
Subsequently, the high-quality dataset was generated based on the PWM algorithm, and the DLRCB-PWM algorithm was proposed to accelerate the convergence speed by applying the DL-based initialization of beamforming vectors.  
\fi
Moreover, the parameters determining the convergence speed of the PWM algorithm and the system performance were analyzed, and then the PWM algorithm was deeply unfolded into a model-driven DL network, i.e., PWM-BFNet, with these parameters trained.
Simulation results verified the PWM-BFNet achieved a significantly high convergence speed and considerable performance improvement, especially at a high total transmit power and a large number of RDARS TEs. Meanwhile, compared to the benchmark schemes, the improved initialization and trainable variables reduced the number of iterations required for convergence, which significantly reduced the computational complexity.
The results in this paper showed the superiority of the RDARS-aided system and illustrated the effectiveness of integrating the model-driven DL into joint beamforming and mode switching design for this system. 
\vspace{-5pt}
\appendices
% \section*{Acknowledgment}
\section*{Appendix A: Proof of Lemma \ref{Lemma: WMMSE equivalent pro}} \label{appendix: proof problem WMMSE}
With \eqref{equ: y_k}, the MSE of $y_k$ is 
\begin{align} \label{equ: MSE in Proof of Proof of Proposition 5}
 e_k &= \mathbb{E}\{(\hat{y}_k - y_k) (\hat{y}_k - y_k)^{H}  \} \nonumber \\
 % & =  \mathbb{E}\{ (( (u^{H}_k{\bf{h}}_k{\bf{f}}_k)-1) s_k  \\& 
 % + \sum_{i\neq k}^{K} u^{H}_k{\bf{h}}_k{\bf{f}}_is_i +u^{H}_k n_k) (( (u^{H}_k{\bf{h}}_k{\bf{f}}_k)-1) s_k \\&
 % + \sum_{i\neq k}^{K} u^{H}_k{\bf{h}}_k{\bf{f}}_is_i +u^{H}_k n_k)^{H} \} \nonumber \\
 & = \mathbb{E}\{\left( (u^{H}_k{\bf{h}}_k{\bf{f}}_k)-1\right) \left( (u^{H}_k{\bf{h}}_k{\bf{f}}_k)-1\right)^H \nonumber  \\&+   u^{H}_k \sum_{i\neq k}^{K} {\bf{h}}_k{\bf{f}}_i{\bf{f}}^{H}_i{\bf{h}}^{H}_i  u_k   + u^{H}_k n_k n^{H}_k u_k   \}  \nonumber \\
 & =1 - u^{H}_k{\bf{h}}_k{\bf{f}}_k - {\bf{f}}^{H}_k {\bf{h}}^{H}_k u_k +  u^{H}_k \sum_{i\neq k}^{K} {\bf{h}}_k{\bf{f}}_i{\bf{f}}^{H}_i{\bf{h}}^{H}_i  u_k  \nonumber\\
 &+  u^{H}_k  u_k \sigma^2_k\sum_{m=1}^{K}\frac{ {\bf{f}}^{H}_m {\bf{f}}_m}{P_{\rm{tot}}}.
\end{align}
An auxiliary term $\sum_{m=1}^{K}({\bf{f}}^{H}_m {\bf{f}}_m)/P_{\rm{tot}}$ is introduced to the last term of \eqref{equ: MSE in Proof of Proof of Proposition 5} to remove the maximum transmit power constraint in \eqref{pro: weighted sum rate}.
Furthermore, given other fixed variables, by letting the first-order derivative of $e_k$ with respect to $\mu_k$ be equal to zero, we have 
\begin{align}\label{equ: derivation of u_k}
  \frac{\partial e_k}{\partial u_k} 
  &= -{\bf{h}}_k {\bf{f}}_k + \sum_{i = 1}^{K} {\bf{h}}_k{\bf{f}}_i{\bf{f}}^{H}_i{\bf{h}}^{H}_k u^{H}_k +  u_k \sigma^2_k = 0.
\end{align}
The optimal solution of ${u_k}$ in \eqref{pro: weighted sum rate MMSE} is given by
\begin{equation}\label{equ: opt u_k in Proof of Proposition 5}
u^{\rm{opt}}_k  = {{J}}_k^{-1} {\bf{h}}_k{\bf{f}}_k,
\end{equation}
where ${{J}}_k= \sum_{i = 1}^{K} {{\bf{h}}_k{\bf{f}}_i{\bf{f}}^{H}_i{\bf{h}}^{H}_k} + \sigma^2_k $.
Similarly, by letting the first-order derivative of $e_k$ with respect to the auxiliary term $\lambda_k$ be equal to zero, we have
\begin{equation}\label{equ: opt labmda in Proof of Proposition 5}
\lambda^{\rm{opt}}_k = e^{-1}_k.
\end{equation}
By substituting $u^{\rm{opt}}_k$ and $\lambda^{\rm{opt}}_k$,  $\forall k\in \mathcal{K}$, into problem \eqref{pro: weighted sum rate MMSE}, we have the following equivalent problem
\begin{subequations}\label{pro: weighted sum rate MMSE equivalent in Proof of Proposition 5}
  \begin{align}
    \mathop {\max }\limits_{{{\bf{F}}},{\bf{\Phi }},{\bf{A}},{\tilde{\bf{A}}}}\;
  & \;\;\sum\limits_{k = 1}^K {{\alpha _k}\log_2 {(e^{\rm{min}}_k) ^{-1}}}
  \\\;\textrm{s.t.}\;
  & \eqref{con: Phi}, \eqref{con: A}, \eqref{con: A tilde}, \eqref{con: A +A tilde},
  \end{align}
\end{subequations}
with $e^{\rm{min}}_k = (1 - {\bf{f}}^{H}_k {\bf{h}}^{H}_k J_k^{-1} {\bf{h}}_k{\bf{f}}_k) ^{-1}$.
By applying the Woodbury matrix identity, we have
\begin{align} \label{equ: e_min in Proof of Proposition 5}
  \log_2{(e^{\rm{min}}_k) ^{-1}} = \log_2{(1 + {\bf{h}}_k{\bf{f}}_k  {\bf{f}}^{H}_k {\bf{h}}^{H}_k \varUpsilon_k^{-1} ) },
\end{align}
with 
\begin{equation}\label{equ: varUpsilon_k}
\varUpsilon_k = J_k^{-1} -  {\bf{h}}_k{\bf{f}}_k  {\bf{f}}^{H}_k {\bf{h}}^{H}_k.
\end{equation}
By combining \eqref{pro: weighted sum rate MMSE equivalent in Proof of Proposition 5} and \eqref{equ: e_min in Proof of Proposition 5}, the proof of Lemma \ref{Lemma: WMMSE equivalent pro} is thus completed.
\vspace{-5pt}
\appendices
\section*{Appendix B: Proof of Lemma~\ref{lemma: Power iteration equivalent}} \label{appendix: proof PI}
With ${\bf{p}}^{H} {\bf{D}}^{'} {\bf{p}} = {\bf{p}}^{H} {{\bf{D}} + \varepsilon {\bf{I}}_{N}} {\bf{p}} = {\bf{p}}^{H} {\bf{D}} {\bf{p}} + (N+1)$, problem \eqref{pro: rankone of D_pie} is thus equivalent to problem \eqref{pro: rankone}.
Then, based on the definition, we have
% \begin{align}
%   \Re\{\left({\bf{p}}^{(q+1)} \right)^{H} {\bf{D}}^{'} {\bf{p}}^{(q)}\} 
%   & = \mathop {\max }\limits_{|p_{n} = 1|}  \Re\{{\bf{p}} ^{H} {\bf{D}}^{'} {\bf{p}}^{q}\}\nonumber \\
%   & \geq \left({\bf{p}}^{(q)} \right)^{H} {\bf{D}}^{'} {\bf{p}}^{(q)}.
% \end{align} 
\begin{align}
  \Re\{({\bf{p}}^{(q+1)} )^{H} {\bf{D}}^{'} {\bf{p}}^{(q)}\}\! =\! \mathop {\max }\limits_{|p_{n} = 1|}\!  \Re\{{\bf{p}} ^{H} {\bf{D}}^{'} {\bf{p}}^{q}\}\nonumber 
\!\! \geq \!\!\left({\bf{p}}^{(q)} \right)^{H} {\bf{D}}^{'} {\bf{p}}^{(q)}.
\end{align} 
When ${\bf{p}}^{(q+1)} \neq  {\bf{p}}^{(q)} $ and the matrix ${\bf{D}}^{'} $ is positive definite, we have 
$({\bf{p}}^{(q+1)} - {\bf{p}}^{(q)} )^{H} {\bf{D}}^{'} ({\bf{p}}^{(q+1)} - {\bf{p}}^{(q)} ) > 0.$
It then follows that
\begin{align}
  ({\bf{p}}^{(q+1)})^{H} {\bf{D}}^{'} {\bf{p}}^{(q+1)} 
  &\!>\! 2\Re\{({\bf{p}}^{(q+1)})^{H} {\bf{D}}^{'} {\bf{p}}^{(q)}\}  - ({\bf{p}}^{(q)} )^{H}{\bf{D}}^{'} {\bf{p}}^{(q)} \nonumber \\
  &\!>\!({\bf{p}}^{(q)} )^{H}{\bf{D}}^{'} {\bf{p}}^{(q)},
\end{align}
which implies that the objective function of problem \eqref{pro: rankone of D_pie} monotonically increases with ${\bf{p}}$.
Furthermore, the objective function of problem \eqref{pro: rankone of D_pie} is upper-bounded by
\begin{align}
  {\bf{p}}^{H}{\bf{D}}^{'}{\bf{p}} &= \sum_{m = 1}^{N+1}\sum_{n=1}^{N+1} p^{*}_{m} {{D}}^{'}_{m,n}p_{n}\leq \sum_{m = 1}^{N+1}\sum_{n=1}^{N+1}|p^{*}_{m} {{D}}^{'}_{m,n}p_{n}| \nonumber\\& \leq
  \sum_{m = 1}^{N+1}\sum_{n=1}^{N+1}| {{D}}^{'}_{m,n}|.
\end{align}
Thus, the power iteration algorithm is guaranteed to converge to a stationary point $\bar{\bf{p}}$, given by
\begin{equation}
  {\bf{D}}^{'} \bar{\bf{p}} = \operatorname*{abs}({\bf{D}}^{'} \bar{\bf{p}}) \odot e^{j \operatorname*{arg}({\bf{D}}^{'} \bar{\bf{p}})} =  \operatorname*{abs}({\bf{D}}^{'} \bar{\bf{p}}) \odot\bar{\bf{p}}.
\end{equation}
Furthermore, we have 
\begin{align}
 & {\bf{p}}^{H} (\operatorname*{diag}( \operatorname*{abs}({\bf{D}}^{'} \bar{\bf{p}})) \!-\! {\bf{D}}^{'}){\bf{p}} 
  \nonumber\\ &= {\bf{p}}^{H}\operatorname*{diag}( \operatorname*{abs}({\bf{D}}^{'} \bar{\bf{p}})) {\bf{p}} - {\bf{p}}^{H}{\bf{D}}^{'}{\bf{p}} \nonumber \\
  & = \sum_{m=1}^{N+1} \left| \sum_{n=1}^{N+1}{{D}}^{'}_{m,n}p_{n}  \right| - \sum_{m=1}^{N+1}\sum_{n=1}^{N+1}{\bf{p}}^{*}_{m}{{D}}^{'}_{m,n}p_{n} \nonumber\\
  &\geq \left|    \sum_{m=1}^{N+1}\sum_{n=1}^{N+1}{\bf{p}}^{*}_{m}{{D}}^{'}_{m,n}p_{n} \right| -  \sum_{m=1}^{N+1}\sum_{n=1}^{N+1}{\bf{p}}^{*}_{m}{{D}}^{'}_{m,n}p_{n} \nonumber \\
  & \geq 0.
\end{align}
Therefore, $\bar{\bf{p}}$ is a local optimum of problem \eqref{pro: rankone} based on the results revealed in  Appendix B of \cite{PI_2014}. This thus completes the proof of Lemma \ref{lemma: Power iteration equivalent}.
\vspace{-5pt}
\section*{Appendix C: Derivation of Auxiliary Parameters} \label{appendix: proof Auxiliary Parameters}

For any ${\bf{a}}$, the last term of the objective function of \eqref{pro: MMSE A} can be expressed as
\begin{align}
\frac{1}{2 \rho} ||{\bf{A}} - \tilde{\bf{A}}\tilde{\bf{A}} ^{H} ||^2_{{F}} 
&= \frac{1}{2 \rho} \operatorname*{Tr}\left( ({\bf{A}}^{H} - \tilde{\bf{A}}\tilde{\bf{A}} ^{H}   ) ({\bf{A}} -\tilde{\bf{A}}\tilde{\bf{A}} ^{H}  )\right)\nonumber \\
& = \frac{1}{2 \rho} \operatorname*{Tr}\left( {\bf{A}}({\bf{I}}_N - 2\tilde{\bf{A}}\tilde{\bf{A}} ^{H}  ) +  
\tilde{\bf{A}}\tilde{\bf{A}} ^{H}
\right) \nonumber \\
& = \frac{1}{2 \rho} ({\bf{r}}^{T}_4  {\bf{a}} + r_5   ),
\end{align}
where ${{\bf{r}}_4}$ and $r_5$ are given in \eqref{equ: auxiliary parameters of A}.
The objective function of \eqref{pro: MMSE A} can be rewritten as
\begin{align} \label{equ: f_1 of A in appendix C}
f_3({\bf{A}}) = \sum_{k=1}^{K} \alpha_k \lambda_k e_k({\bf{a}}) + \frac{1}{2 \rho} ||{\bf{A}} - \tilde{\bf{A}}\tilde{\bf{A}}^{H} ||^2_{{F}}.
\end{align}
with $e_k({\bf{a}})$ given by
% \begin{align}
% {e_k}({\bf{a}}) 
%  &= u_k^H{{\bm{\varphi }}^H} \operatorname{diag}({{\bf{H}}_{r,k}}{{\bf{w}}_{b,k}}){\bf{a}} + {u_k}{{\bf{a}}^T} \operatorname{diag}({\bf{w}}_{b,k}^H{{\bf{H}}_{r,k}}^H){\bm{\varphi }} \nonumber\\
%  &- {{\bf{a}}^T}u_k^H{u_k}{{\bf{\Phi }}^H}{{\bf{H}}_{r,k}}\sum\limits_{m = 1}^K {{{\bf{w}}_{b,m}}{\bf{w}}_{b,m}^H} {\bf{H}}_{r,k}^H{\bm{\varphi }} \nonumber\\
% &
%  - u_k^H{u_k}{{\bm{\varphi }}^H}{{\bf{H}}_{r,k}}\sum\limits_{m = 1}^K {{{\bf{w}}_{b,m}}{\bf{w}}_{b,m}^H} {\bf{H}}_{r,k}^H{\bf{\Phi a}} \nonumber\\
%  &+ u_k^H{u_k}{{\bf{a}}^T}{{\bf{\Phi }}^H}{{\bf{H}}_{r,k}}\sum\limits_{m = 1}^K {{{\bf{w}}_{b,m}}{\bf{w}}_{b,m}^H} {\bf{H}}_{r,k}^H{\bf{\Phi a}}\nonumber\\
% & - u_k^H{u_k}\sum\limits_{m = 1}^K {{\bf{h}}_{r,k}^H{\bf{\tilde A}}{{\bf{w}}_{r,m}}{\bf{w}}_{b,m}^H} {{\bf{H}}_{r,k}}^H{\bf{\Phi a}} 
%  \nonumber\\
%  &- u_k^H{u_k}{{\bf{a}}^T}{{\bf{\Phi }}^H}{{\bf{H}}_{r,k}}\sum\limits_{m = 1}^K {{{\bf{w}}_{b,m}}{\bf{w}}_{r,m}^H} {{{\bf{\tilde A}}}^H}{{\bf{h}}_{r,k}}.
% \end{align}
${e_k}({\bf{a}}) = u_k^H{{\bm{\varphi }}^H} \operatorname{diag}({{\bf{H}}_{r,k}}{{\bf{w}}_{b,k}}){\bf{a}} + {u_k}{{\bf{a}}^T} \operatorname{diag}({\bf{w}}_{b,k}^H{{\bf{H}}_{r,k}}^H){\bm{\varphi }}- {{\bf{a}}^T}u_k^H{u_k}{{\bf{\Phi }}^H}{{\bf{H}}_{r,k}}\sum\limits_{m = 1}^K {{{\bf{w}}_{b,m}}{\bf{w}}_{b,m}^H} {\bf{H}}_{r,k}^H{\bm{\varphi }} 
 - u_k^H{u_k}{{\bm{\varphi }}^H}{{\bf{H}}_{r,k}}\sum\limits_{m = 1}^K {{{\bf{w}}_{b,m}}{\bf{w}}_{b,m}^H} {\bf{H}}_{r,k}^H{\bf{\Phi a}} + u_k^H{u_k}{{\bf{a}}^T}{{\bf{\Phi }}^H}{{\bf{H}}_{r,k}}\sum\limits_{m = 1}^K {{{\bf{w}}_{b,m}}{\bf{w}}_{b,m}^H} {\bf{H}}_{r,k}^H{\bf{\Phi a}} - u_k^H{u_k}\sum\limits_{m = 1}^K {{\bf{h}}_{r,k}^H{\bf{\tilde A}}{{\bf{w}}_{r,m}}{\bf{w}}_{b,m}^H} {{\bf{H}}_{r,k}}^H{\bf{\Phi a}} 
- u_k^H{u_k}{{\bf{a}}^T}{{\bf{\Phi }}^H}{{\bf{H}}_{r,k}}\sum\limits_{m = 1}^K {{{\bf{w}}_{b,m}}{\bf{w}}_{r,m}^H} {{{\bf{\tilde A}}}^H}{{\bf{h}}_{r,k}}$.
Therefore, we have 
${f_3}({\bf{A}}) 
= {\bf{r}}_1^H{\bf{a}} + {{\bf{a}}^T}{{\bf{r}}_1} + {{\bf{a}}^T}{{\bf{r}}_2} + {\bf{r}}_2^H{\bf{a}} + {{\bf{a}}^T}{{\bf{R}}_1}{\bf{a}} + {\bf{r}}_3^H{\bf{a}}+ {{\bf{a}}^T}{{\bf{r}}_3} + \frac{1}{{2\rho }}({\bf{r}}_4^T{\bf{a}} + {{{r}}_5})$,
where ${\bf{r}}_1$, ${{\bf{r}}_2}$,  ${{\bf{r}}_3}$,  ${{\bf{r}}_4}$, and ${{{r}}_5}$ are given in \eqref{equ: auxiliary parameters of A}.

Similarly, the objective function of \eqref{pro: MMSE A + A tilde} can be expressed as
\begin{align} \label{equ: f_2 of A tilde in appendix C}
  f_4(\tilde{{\bf{A}}}) = \sum_{k=1}^{K} \alpha_k \lambda_k e_k(\tilde{{\bf{a}}}) + \frac{1}{2 \rho} ||{\bf{A}} - \tilde{\bf{A}}\tilde{\bf{A}}^{H} ||^2_{{F}},
\end{align}
 with $e_k(\tilde{{\bf{a}}})$ given by
\begin{align}
{e_k}(\tilde{{\bf{a}}}) 
% &= u_k^H{u_k}\sum\limits_{m = 1}^K {{{\bf{h}}^H_{r,k}}{\bf{\tilde A}}{{\bf{w}}_{r,m}}{\bf{w}}_{b,m}^H} {\bf{H}}_{r,k}^H{\bf{(}}{{\bf{I}}_N}{\bf{ - A)\bm{\varphi} }} 
% \nonumber\\
% &- u_k^H{\bf{h}}_{r,k}^H{\bf{\tilde A}}{{\bf{w}}_{r,m}} - {\bf{w}}_{r,m}^H{{{\bf{\tilde A}}}^H}{{\bf{h}}_{r,k}}{u_k} \nonumber\\
% &
% + u_k^H{u_k}{{\bm{\varphi }}^H}{\bf{(}}{{\bf{I}}_N}{\bf{ - A)}}{{\bf{H}}_{r,k}}\sum\limits_{m = 1}^K {{{\bf{w}}_{b,m}}{\bf{w}}_{r,m}^H} {{{\bf{\tilde A}}}^H}{{\bf{h}}_{r,k}} \nonumber\\
% &+ u_k^H{u_k}\sum\limits_{m = 1}^K {{\bf{h}}_{r,k}^H{\bf{\tilde A}}{{\bf{w}}_{r,m}}{\bf{w}}_{r,m}^H} {{{\bf{\tilde A}}}^H}{{\bf{h}}_{r,k}} \nonumber\\
& = u_k^H{u_k}{{\bf{h}}^H_{r,k}}{\bf{\tilde A}}\sum\limits_{m = 1}^K {{{\bf{w}}_{r,m}}{\bf{w}}_{b,m}^H} {\bf{H}}_{r,k}^H{\bf{(}}{{\bf{I}}_N}{\bf{ - A){\bm{\varphi }} }} \nonumber\\
&- u_k^H{\bf{h}}_{r,k}^H{\bf{\tilde A}}{{\bf{w}}_{r,k}} - {\bf{w}}_{r,k}^H{{{\bf{\tilde A}}}^H}{{\bf{h}}_{r,k}}{u_k} \nonumber\\
&+ u_k^H{u_k}{{\bm{\varphi }}^H}{\bf{(}}{{\bf{I}}_N}{\bf{ - A)}}{{\bf{H}}_{r,k}}\sum\limits_{m = 1}^K {{{\bf{w}}_{b,m}}{\bf{w}}_{r,m}^H} {{{\bf{\tilde A}}}^H}{{\bf{h}}_{r,k}}
\nonumber\\
& + u_k^H{u_k}{\bf{h}}_{r,k}^H{\bf{\tilde A}}\sum\limits_{m = 1}^K {{{\bf{w}}_{r,m}}{\bf{w}}_{r,m}^H} {{{\bf{\tilde A}}}^H}{{\bf{h}}_{r,k}}.
\end{align}
Moreover,  the last term of \eqref{equ: f_2 of A tilde in appendix C} can be expressed as
\begin{align}
  \frac{1}{2 \rho} ||{\bf{A}} - \tilde{\bf{A}}\tilde{\bf{A}} ^{H} ||^2_{{F}} 
  % &= \frac{1}{2 \rho} \operatorname*{Tr}\left( ({\bf{A}}^{H} - \tilde{\bf{A}}\tilde{\bf{A}} ^{H}   ) ({\bf{A}} -\tilde{\bf{A}}\tilde{\bf{A}} ^{H}  )\right)\nonumber \\
  & = \frac{1}{2 \rho} \left( \operatorname*{Tr}({\bf{A}} + \tilde{\bf{A}}\tilde{\bf{A}} ^{H}) - 2\operatorname*{Tr}({\bf{A}}\tilde{\bf{A}}\tilde{\bf{A}} ^{H} )\right)
   \nonumber \\
  & = \frac{1}{2 \rho} (a - {\tilde{\bf{a}}^{H}}({\bf{I}}_{a}\otimes {\bf{A}})\tilde{\bf{a}} ).
\end{align}
Therefore, we have \eqref{pro: A}, \eqref{equ: auxiliary parameters of A}, and \eqref{pro: tilde A}, which completes the proof.
\vspace{-5pt}
\bibliographystyle{IEEEtran}
\bibliography{IEEEabrv, Reference}
\end{document}